%% file: SPAWC2021.tex
\documentclass[10pt,journal]{IEEEtran} 

\usepackage[utf8]{inputenc}

\usepackage{fourier}
\usepackage{tikz}
\usepackage{amsmath,amssymb}
\usepackage{dsfont}
\usepackage{bm}
\usepackage{pgfplots}
\usepackage{todonotes}
\usepackage{tikz}
\usepackage{comment}
\usepackage{algorithm}
\usepackage{algpseudocode}
\usepackage{amsthm}
\usepackage{subfig}
\newtheorem{prop}{Proposition}
\newtheorem{lemma}{Lemma}
\DeclareMathOperator{\diag}{diag}
\DeclareMathOperator{\bdiag}{bdiag}

\pgfplotsset{compat=1.16}

\def \A {\mathbf A}
\def \B {\mathbf B}
\def \C {\mathbf C}

\def \G {\mathbf G}
\def \Z {\mathbf Z}
\def \K {\mathbf K}
\def \D {\mathbf D}
\def \E {\mathbf E}

\def \X {\mathbf X}

\def \M {\mathbf M}

\def \U {\mathbf U}
\def \V {\mathbf V}
\def \S {\mathbf S}
\def \I {\mathbf I}

\def \P {\mathbf P}
\def \U {\mathbf U}
\def \w {\mathbf w}

\def \s {\mathbf s}

\def \x {\mathbf x}
\def \y {\mathbf y}
\def \z {\mathbf z}

\def \e {\mathbf e}
\def \u {\mathbf u}

\def \w {\mathbf w}

\def \h {\mathbf h}

\def \btheta {\boldsymbol \theta}

\def \blambda {\boldsymbol \lambda}
\def \bGamma {\boldsymbol \Gamma}

\def \Ka {K_{\mathrm{a}}}

\def \vec {\mathrm{vec}}
\def \LLR {\mathrm{LLR}}
\def \GLRT {\mathrm{GLRT}}

\def \SNR  {\mathrm{SNR}}

\def \MSE {\mathrm{MSE}}

\usepackage{calrsfs}
\DeclareMathAlphabet{\pazocal}{OMS}{zplm}{m}{n}
\makeatletter
\def\input@path{{figs/}{./}}
\makeatother

\title{Tensor Decomposition Bounds for TBM-Based Massive Access}
\author{Alexis Decurninge, Ingmar Land and Maxime Guillaud%
\thanks{The authors are with the Advanced Wireless Technology Lab, Paris Research Center, Huawei Technologies France~(e-mail: \texttt{\{alexis.decurninge,ingmar.land,maxime.guillaud\}}\newline \texttt{@huawei.com}).}%
}

\begin{document}
\maketitle

\begin{abstract}
Tensor-based modulation (TBM) has been proposed in the context of unsourced random access for massive uplink communication.
In this modulation, transmitters encode data as rank-1 tensors, with factors from a discrete vector constellation.
This construction allows to split the multi-user receiver into a user separation step based on a low-rank tensor decomposition, and independent single-user demappers.
In this paper, we analyze the limits of the tensor decomposition using Cram\'er-Rao bounds, providing bounds on the accuracy of the estimated factors.
These bounds are shown by simulation to be tight at high SNR.
We introduce an approximate perturbation model for the output of the tensor decomposition, which facilitates the computation of the log-likelihood ratios (LLR) of the transmitted bits, and provides an approximate achievable bound for the finite-length error probability.
Combining TBM with classical forward error correction coding schemes such as polar codes, we use the approximate LLR to derive soft-decision decoders showing a gain over hard-decision decoders at low SNR.
\end{abstract}

\section{Introduction}
\subsection{Context}

The increasing traffic demand of users with sporadic activity, in particular for machine-to-machine communications, translates into the need for low-overhead, spectrally efficient waveforms for grant-free random access.
\emph{Unsourced} random access schemes \cite{polyanskyi17} constitute a step in that direction, by assuming that all transmitters use the same codebook, i.e. the identities of the active transmitters are not required by the multi-user decoder (although they can be embedded in the payloads).
In this case, the receiver decodes a list of messages up to a permutation, regardless of the active users' identities.

%However, the performances of the activity detection step are crucially dependent on the number of potential transmitters when their number become large.
Traditional grant-free schemes use (possibly non-orthogonal) pilot sequences in order to both identify active users and estimate their channels \cite{chen18}.
The receiver subsequently performs coherent detection in order to recover the transmitted symbols.
This approach can also be used in the unsourced approach by using a part of the payload to determine the pilot sequence \cite{fengler20}.
On the other hand, many unsourced access schemes avoiding the pilot/data division have also been proposed in particular for the SISO AWGN case \cite{fengler19a} or SISO with Rayleigh fading \cite{kowshik19}.
As for the MIMO setting, a scheme based on compressed sensing has been proposed in \cite{fengler19} and a tensor-based approach \cite{decurninge20}. 
This article focuses on the performance analysis of the latter.

\subsection{System Model and Modulation}

Let us consider an uplink transmission between single-antenna transmitters and an $N$-antennas receiver within a block of $T$ channel uses. 
We assume that $\Ka$ users are active for the duration of the considered block (the total number of users in the system plays no role in the analysis of unsourced schemes) and simultaneously transmit a payload of $B$ information bits each, encoded using a forward error correction code.
Second, we assume a block-fading model whereby the channel state remains constant over the considered block of length $T$, and is a priori unknown to both the transmitters and the receiver.
Let $\h_k\in \mathbb{C}^{N}$ denote the channel from user $k$ to the $N$ receive antennas.
Finally, we assume block synchronization between the transmitters and the receiver.
%Since the number $K$ of potential users (accounting for both active and inactive transmitters) is large, we assume that all users use the same constellation $\mathcal{C}=\{\c_1,\dots,\c_{2^B}\}$ containing $2^B$ elements.
We consider the tensor-based modulation (TBM) from \cite{decurninge20}, i.e. the codebook $\mathcal{C}$ is composed of Kronecker products of elements from vector sub-constellations $\mathcal{C}_i\subset \mathbb{C}^{T_i}$, namely
\begin{equation}
\mathcal{C} = \Big\{\x_1\otimes\dots\otimes\x_d: \ \x_1\in\mathcal{C}_1, \ \dots \ , \ \x_d\in\mathcal{C}_d\Big\},
\end{equation}
where $T_i>1$ and $\prod_{i} T_i = T$.
We assume furthermore that the constellation elements have equal norm $\|\s\|^2=T$ for any $\s\in\mathcal{C}\subset \mathbb{C}^T$.
%Note that the Kronecker structure is similar to the construction in \cite{freitas18} proposed in the point-to-point case.
Information is encoded in $\x_1\ldots \x_d$, using independent vector modulations in each mode of the tensor.
Let $\x_{1,k}\otimes \dots\otimes \x_{d,k}$ denote the vector symbol sent by the $k$-th user; considering an additive Gaussian i.i.d. noise vector $\w$, the signal received by the $N$ antennas is modeled by
\begin{equation}
\label{eq:receive}
\y = \sum_{k=1}^{\Ka} \x_{1,k}\otimes \dots\otimes \x_{d,k} \otimes \h_k + \w \ \in \mathbb{C}^{TN}.
\end{equation}
%where $\w\in\mathbb{C}^{TN}$ is a Gaussian i.i.d. noise vector.

\subsection{Receiver Architecture}
In the sequel, we assume that $\Ka$ is known to the receiver, and focus on the implementation of a soft demapper.
Since the enumeration of the joint (across the active users) constellation $\mathcal{C}^{\Ka}$ is typically not computationally feasible, we consider the following functional split of the demapper:
\begin{itemize}
\item A low-rank tensor approximation step, whereby the received signal is approximated by a rank-$\Ka$ tensor, thus jointly estimating the channel and the symbols. At this stage, the discrete nature of $\mathcal{C}$ is ignored and the $\x_{i,k}$ are treated as continuous parameters. This corresponds to the following ML estimator:
\begin{equation}
\label{eq:cpd}
\hspace{-0.2cm}
\small
\{\hat\z_{i,k},\hat\h_k\}_{\substack{i\in[d], \\ k\in[\Ka]}} = \arg\min_{\substack{\h_k\in\mathbb{C}^{N}\\ \z_{i,k}\in\mathbb{C}^{T_i}}} \left\|\y - \sum_{k=1}^{\Ka} \z_{1,k}\otimes \dots\otimes \z_{d,k} \otimes \h_k\right\|^2.
\end{equation}
Note that this step effectively performs user separation, thanks to the property that each user is associated to a rank-1 tensor \cite{decurninge20}.
In order to guarantee the unicity of a solution to \eqref{eq:cpd}, we may impose constraints on the parameters $\z_{i,k}$ -- as detailed in Section~\ref{section_CRB}.
\item A single-user soft demapper, where the log-likelihood ratios (LLRs) are computed independently for each user, based on the outputs of the tensor decomposition computed at the previous step, and used as inputs to a soft-decision decoder.
Soft demapping requires to have access to the distribution of the symbols conditioned on these variables, which we denote by $p(\x_{i,k}|\{\hat\z_{j,k}\},\hat\h_k)$.
\begin{comment}
$\Ka$ discrete ML estimators based on the tensor outputs take place, i.e. for all $k = 1\ldots \Ka$,
\begin{equation}
\label{eq:single-user}
\hat\x_1,\ldots,\hat\x_d= \arg\max_{\x_1, \ldots, \x_d\in\mathcal{C}} p(\x_1\otimes\dots\otimes\x_d|\hat\z_{1,k},\dots,\hat\z_{d,k})
%\hat\x_{1,k}\otimes\dots\otimes\hat\x_{d,k}= \arg\min_{\s\in\mathcal{C}} \|\s - \hat\z_{1,k}\otimes\dots\otimes\hat\z_{d,k}\|^2.
\end{equation}
\end{comment}
\end{itemize}
Solving \eqref{eq:cpd} in the noise-free case is a well-studied problem (known as the canonical polyadic decomposition, or CPD) in tensor algebra.
Conversely, in the presence of noise, there is no established perturbation model for the solution of the CPD.
Consequently, $p(\x_{i,k}|\{\hat\z_{j,k}\},\hat\h_k)$ is not known in closed form.
The general objective of this paper is to propose a simple approximation of $p(\x_{i,k}|\hat\z_{i,k},\hat\h_k)$ that is suitable for use in a soft demapper, and therefore we neglect the dependency on $\hat\z_{j,k}, j\neq i$.

We first analyze the performance limits of the tensor decomposition seen as a statistical estimation in Section~\ref{sec:cpd} and then propose asymptotic approximations of the LLR and the achievable rate in Section~\ref{sec:llr}.
Finally, we illustrate our analysis through numerical simulations in Section~\ref{sec:simul}.

\section{Theoretical Performance of the Approximate CPD}
\label{sec:cpd}
\subsection{Cram\'er-Rao (CR) Bounds}
\label{section_CRB}
For the sake of notational simplicity, let $p=d+1$ and $\x_{p,k}=\h_k$ for any $1\leq k\leq \Ka$.
Consider the estimation of the factors $(\x_{i,k})_{1\leq i\leq p, 1\leq k\leq \Ka}$ in \eqref{eq:receive}, and let $\hat\z_{i,k}$ denote an estimate of $\x_{i,k}$. We further denote $\hat\Z_i = (\hat\z_{i,1},\dots,\hat\z_{i,\Ka})$ and $\X_i = (\x_{i,1},\dots,\x_{i,\Ka})$.
Therefore,  $\hat\btheta=(\vec(\hat\Z_1)^T,\dots,\vec(\hat\Z_p)^T)^T$ is an estimate of $\btheta=(\vec(\X_1)^T,\dots,\vec(\X_p)^T)^T$.
We now derive a lower bound on the variance of $\hat\btheta$ using CR analysis.

The difficulty of deriving such bounds for the tensor parameters lies in the fact that there are multiple ways to represent the same rank-1 tensor as the product $\x_{1,k}\otimes \dots \otimes \x_{p,k}$ (see e.g. \cite{decurninge20}).
This directly translates into scalar ambiguities on the parameters $\{\x_{i,k}\}$: the elements of the Kronecker product can be multiplied by arbitrary complex scalar coefficients, provided that their product is equal to 1.
In the noise-free case, these ambiguities are not detrimental to detection when the $\mathcal{C}_i$ are defined using vector constellations adapted to non-coherent channels, as in the case of Grassmannian codebooks \cite{decurninge20}.
However, they must be accounted for, and they complicate the derivation of the CR bounds when detection is considered in the presence of noise; attempting to derive CR bounds while ignoring the representation ambiguity will yield ill-conditioned results.
They can be resolved by introducing constraints on the sub-constellations $\mathcal{C}_i$, such as $\|\x_{i,k}\|^2=T_i$ and $\text{Im}(\x_{i,k}^H\x_{i,k}^0)=0$ where $\x_{i,k}^0$ is a reference vector for any $k$ and any $1\leq i\leq p-1$.
Recall that the $p$-th mode corresponds to the fading process ($\x_{p,k}=\h_k$), and therefore we make no assumption regarding the norm of $\x_{p,k}$.
Note that the norm constraints induce that the total energy used by each transmitted codebook vector is equal to $T$, i.e. $\|\x_{1,k}\otimes\dots\otimes\x_{d,k}\|^2=T$.
%Furthermore, there is a natural norm constraint coming from transmit power; here, we assume a unit power of energy per channel access at each transmitter, i.e.  $\|\x_{1,k}\otimes\dots\otimes\x_{d,k}\|^2=T$.
%Note also that imposing such constraints should then be respected by the codebook design.

%These additional constraints are sufficient to solve the scaling ambiguity. 
%Note that these constraints are imposed on the true parameters.
CR bounds can be established by assuming identical constraints on the set of estimated parameters $\hat{\btheta}$.
In particular, we will consider the constraints
\begin{equation}
\text{for any } \begin{array}{l}1\leq k\leq \Ka\\ 1\leq i\leq p-1\end{array}, \quad 
  \left\{\begin{array}{ll}
\|\hat{\z}_{i,k}\|^2 = \|\x_{i,k}\|^2 = T_i\\
\text{Im}(\x_{i,k}^H\hat{\z}_{i,k}) = 0
\end{array}\right.
\label{eq:constraints}
\end{equation}
wherein the first constraint deals with the norm ambiguity and the second with the phase ambiguity. 
Note that we chose a phase constraint that requires the knowledge on the true parameter $\x_{i,k}$ since it leads to more tractable bounds.

%In practice, we will resolve this indeterminacy at the single-user demapper stage (see Section~\ref{sec:equivalent_model}).
%\todo[inline]{We will discuss this point below: be more specific -> (AD) I propose the above (MG) I made a reference to non-coherent modulations above instead}
Note that there remains a sign ambiguity on $\{\x_{i,k}\}_{1\leq i\leq p-1}$ and a permutation ambiguity over the indices $1\leq k\leq \Ka$ that can be resolved with additional inequality constraints which have negligible impact on CR bounds \cite{liu_sidiropoulos_01}.
%Note furthermore that there remains a sign ambiguity that can be resolved by considering e.g. the inequality constraint $\text{Re}(\x_{i,k}^H\hat{\z}_{i,k}) > 0$.
%The permutation ambiguity can be resolved for example by imposing another inequality constraint such as $\|\hat\btheta_1\|> \|\hat\btheta_2\|>\dots>\|\hat\btheta_K\|$.
%We will consider that the inequality constraints have no impact on the CR bound (see e.g. \cite[Section II.C.]{gorman90} justifying this statement under mild conditions).
CR bounds for the tensor decomposition have been derived for unbiased estimators e.g. in \cite{liu_sidiropoulos_01} or \cite{tichavsky13} with similar constraints. 
However, the norm constraint in \eqref{eq:constraints} makes the unbiased assumption unsuitable, since for any estimator $\hat\z_{i,k}$ satisfying the norm constraint and whose distribution admits a density on the sphere of radius $\sqrt{T_i}$, it holds that $\big|\mathbb{E}[\hat\z_{i,k}^H\x_{i,k}]\big|<T_i$, i.e. $\mathbb{E}[\hat\z_{i,k}]$ is strictly inside the sphere while $\hat\z_{i,k}$ is on the sphere, therefore the estimate $\hat\z_{i,k}$ is necessarily biased. 
We will hence consider the class of \emph{biased} estimators of $\x_{i,k}$ satisfying the constraints \eqref{eq:constraints}.
Furthermore, for the sake of simplicity of exposition, we will restrict ourselves to a bias parameterized by scalars $\alpha_{i,k}\in[0,1]$ for $1\leq i\leq d$ through
\begin{equation}
\mathbb{E}[\hat\z_{i,k}] =  \alpha_{i,k}\x_{i,k} \text{\quad for any $1\leq i\leq d$, $1\leq k\leq \Ka$}.
\label{eq:bias}
\end{equation}
Note that \eqref{eq:bias} only concerns the information-bearing modes ($1\leq i\leq d$); conversely, we seek an unbiased estimator for $\x_{p,k}=\h_k$ since it is not assumed to lie on a sphere.

We will bound the accuracy of the estimator $\hat\z_{i,k}$ for $1\leq i\leq d$ in terms of a normalized variance $\xi_{i,k}$ defined as
\begin{equation}
\xi_{i,k} = \frac{\mathbb{E}\big[\|\hat{\z}_{i,k} - \alpha_{i,k}\x_{i,k}\|^2\big]}{T_i\alpha_{i,k}^2}.
\label{eq:xidef}
\end{equation}

In the following proposition, we reformulate the results of \cite{tichavsky13} in the context of biased estimators in order to derive a lower bound on $\xi_{i,k}$.
%These bounds are derived as lower bounds of the CR bound.
%Even though the derived lower bounds may not be achievable, we think that these lower bounds are meaningful, especially in the low noise regime.

\begin{prop}
\label{prop:crb}
Let $\s_{-i,k}$ denote a partial Kronecker product of the sub-constellation vectors at user $k$ where mode $i$ is omitted, i.e. for $1\leq i\leq p$,
\begin{equation}
\s_{-i,k} = \x_{1,k}\otimes \dots\otimes \x_{i-1,k}\otimes \x_{i+1,k}\otimes\dots\x_{p,k}.
\label{eq:Sdef}
\end{equation}
Define also
\begin{equation}
\S_{-i} = (\s_{-i,1},\dots,\s_{-i,\Ka})\in\mathbb{C}^{\prod_{m\neq i}T_m\times \Ka}  \text{\quad and\quad}  \bGamma_{i} = \S_{-i}^H\S_{-i}.
\end{equation}
Furthermore, let $\S_{-i,-k}$ denote the matrix $\S_{-i}$ with the $k$-th column removed, and $P^{\perp}_{\S_{-i,-k}}$ the projection matrix on the subspace orthogonal to the column span of $\S_{-i,-k}$.
Similarly, we denote by $P_{\x_{i,k}}^{\perp}\in\mathbb{C}^{T_i\times (T_i-1)}$ an orthonormal basis of vectors orthogonal to $\x_{i,k}$.
We assume that $\bGamma_{i}$ is an invertible $\Ka\times \Ka$ matrix, and let $\e_k$ denote the $k$-th vector of the canonical basis and $\U_i$ the orthonormal basis spanned by the $\Ka-1$ vectors $(\bGamma_{i}^{-1/2}\e_{k'}\otimes \x_{i,k'})_{k'\neq k}$ and $\V_i = \frac{\bGamma_{i}^{-1/2}\e_k}{\|\bGamma_{i}^{-1/2}\e_k\|}\otimes P_{\x_{i,k}}^{\perp}$.
Then, for any $1\leq i\leq p-1$, any estimator satisfying \eqref{eq:constraints} and \eqref{eq:bias} satisfies
\begin{equation}
\xi_{i,k} \geq \frac{\left(T_i-1-\|\U_i^H\V_i\|^2\right)\sigma^2}{T_i\|P^{\perp}_{\S_{-i,-k}}\s_{-i,k}\|^2}.
\label{eq:cr_variance}
\end{equation}
In the particular case $\Ka=1$, this lower bound simplifies into 
\begin{equation}
\xi_{i,k} \geq \frac{\sigma^2(T_i-1)}{T_i\prod_{j\neq i}\|\x_{j,k}\|^2}.
\end{equation}
\end{prop}
\begin{proof}
The result is obtained by lower-bounding the constrained CRB. See Appendix~\ref{app:crb} for details.
\end{proof}

Note that the result for $\Ka=1$ corresponds to the CR bound (see \cite[eq. 40]{tichavsky13}).
The assumption that $\bGamma_{i}$ is invertible in Prop.~\ref{prop:crb} is reasonable for large blocksizes and payloads, assuming that the vector symbols are independent across users.
In fact, if the signals generated by different users have low scalar product, $\bGamma_{i}^{-1}$ is close to a diagonal matrix and the term $\|\U_i^H\V_i\|^2$ in \eqref{eq:cr_variance} becomes negligible.
This leads to the following, more explicit bound for $\x_{i,k}$:
\begin{prop}
\label{prop:crb2}
Define $0<\eta_i^-<\eta_i^+$ such that\footnote{We use $\A\preceq\B$ to denote that $\B-\A$ is a positive semi-definite matrix.}
\begin{equation}
%\eta_i^-\left(\begin{array}{ccc}(\bGamma_{i}^{-1})_{11} & 0 & 0 \\ 0 & \ddots & 0 \\ 0 & 0 & (\bGamma_{i}^{-1})_{\Ka\Ka} \end{array}\right) \preceq \bGamma_{i}^{-1}\preceq \eta_i^+\left(\begin{array}{ccc}(\bGamma_{i}^{-1})_{11} & 0 & 0 \\ 0 & \ddots & 0 \\ 0 & 0 & (\bGamma_{i}^{-1})_{\Ka\Ka} \end{array}\right).
\eta_i^-\diag((\bGamma_{i}^{-1})_{kk})_{k=1\dots\Ka}\preceq \bGamma_{i}^{-1}\preceq \eta_i^+\diag((\bGamma_{i}^{-1})_{kk})_{k=1\dots\Ka},
\end{equation}
and let us denote
\begin{equation}
\xi_{i,k}^* = \frac{\sigma^2(T_i-1)}{T_i\|P^{\perp}_{\S_{-i,-k}}\s_{-i,k}\|^2} \left(1 - \frac{(\eta_i^+)^2-1}{(T_i-1)\eta_i^-}\right).
\label{eq:xistar}
\end{equation}
Then, for any estimator satisfying \eqref{eq:constraints} and \eqref{eq:bias}, we have 
\begin{equation}
\xi_{i,k} \geq \xi_{i,k}^* \text{\quad for $1\leq i\leq p-1$}.
\end{equation}
\end{prop}
\begin{proof}
See Appendix~\ref{app:crb2}.
\end{proof}
Considering \eqref{eq:xidef}, Proposition~\ref{prop:crb2} indicates that $\alpha_{i,k}^2\xi_{i,k}^*T_i$ is a lower bound on the variance of any estimator of $\x_{i,k}$.
%In the case of linearly independent signals across the users, $\bGamma_{i}$ is diagonal and $\eta_i^- = \eta_i^+ =1$ for $1\leq i\leq p-1$.
Moreover, if $T\gg\Ka$, $\bGamma_{i}$ is diagonal and $\eta_i^- = \eta_i^+ =1$ for $1\leq i\leq p-1$.
Furthermore, $P^{\perp}_{\S_{-i,-k}}$ is independent from $\s_{-i,k}$.
Assuming that the factors $\x_{i,k}$ of each user  are isotropically drawn\footnote{Grassmannian codebooks designed with a typical max-min distance criterion are asymptotically equivalent to a uniform random variable on a sphere.}, we can approximate the distribution of $P^{\perp}_{\S_{-i,-k}}$ by an isotropic distribution on the space of projectors on a subspace of size $\prod_{j\neq i}T_j - \Ka +1$.
As a consequence, we have $\|P^{\perp}_{\S_{-i,-k}}\s_{-i,k}\|^2\approx \frac{\prod_{j\neq i}T_j - \Ka +1}{\prod_{j\neq i} T_j}\prod_{j\neq i}\|\x_{j,k}\|^2$, which yields for $1\leq i\leq d$ using \eqref{eq:constraints}
\begin{eqnarray}
\xi_{i,k}^* &\approx&  \frac{\prod_{j\neq i} T_j}{\prod_{j\neq i}T_j - \Ka +1}\frac{(T_i-1) \sigma^2}{\prod_{j\neq i}\|\x_{j,k}\|^2} \\
\xi_{i,k}^* &\approx& \frac{(T_i-1) \sigma^2}{TN - T_i(\Ka -1)}\frac{N}{\|\h_{k}\|^2}. %\frac{(T_i-1)\sigma^2}{\prod_{j\neq i}T_j -K+1}\frac{\prod_{j\neq i}T_j}{{\prod_{j\neq i}\|\x_j\|^2}}
\label{eq:mse_approx}
\end{eqnarray}
Note that \eqref{eq:mse_approx} only depends on the norm of the last mode of the tensor (i.e. the channels) since it is not concerned by \eqref{eq:constraints}.
Interestingly, \eqref{eq:mse_approx} indicates that whenever $\prod_{j\neq i}T_j \gg \Ka$, the number of active users $\Ka$ has a negligible influence on the approximate variance of $\xi_{i,k}^*$, which denotes a low-interference regime.

\subsection{Mean Square Error (MSE) for the constrained estimators}
Note that the performance metric $\xi_{i,k}$ in \eqref{eq:xidef} depends on the bias, which is unknown a priori, and therefore is not directly exploitable.
A more practical metric is the Mean Square Error (MSE): 
\begin{equation}
\MSE_{i,k} = \frac{1}{T_i} \mathbb{E}\big[\|\hat\z_{i,k} - \x_{i,k}\|^2\big].
\end{equation}
The following lemma shows that the constraint \eqref{eq:constraints} introduces a coupling between the estimator bias and the MSE.
\begin{lemma}
\label{prop:mse}
For any estimator satisfying \eqref{eq:constraints} and \eqref{eq:bias}, it holds 
\begin{equation}
\label{eq:alpha}
\alpha_{i,k}^2 = \frac{1}{1+ \xi_{i,k}}.
\end{equation}
In other words, in this model, the bias $\alpha_{i,k}$ is not a free parameter that can be optimized to minimize the MSE. 
With the hypothesis of Proposition~\ref{prop:crb2}, the MSE is lower bounded by 
\begin{equation}
\label{eq:mse_lb}
\frac{1}{T_i}\mathbb{E}\big[\|\hat{\z}_{i,k} - \x_{i,k}\|^2\big] \geq 2\left(1 - \frac{1}{(1+\xi_{i,k}^*)^{1/2}}\right).
\end{equation}
\end{lemma}
\begin{proof}
See Appendix~\ref{app:bias_variance}.
\end{proof}

\subsection{MSE asymptotics of Approximate Message Passing (AMP)-based low-rank tensor decomposition}

In \cite{kadmon18}, the Bayesian maximum a posteriori estimator of the known-rank tensor decomposition similar to \eqref{eq:cpd} is approximated using an AMP algorithm with the additional assumption that each element of the modes, i.e. the $x_{i,k}(t)$ are distributed according to a Gaussian distribution of known mean $\mu$ and variance $\sigma_0^2$. 
Even though the results of \cite{kadmon18} are derived for real-valued tensors, we conjecture that the asymptotic behavior of the AMP algorithm is similar for complex-valued tensors.
We will restrict ourselves to the case $\mu=0$ and $\sigma_0^2=1$ in the following. 
These two hypotheses slightly differ from our assumptions: in our notations they imply $\mathbb{E}[\|\x_{i,k}\|^2]=T_i$ for all $i=1\ldots p$, which is weaker than the energy constraint $\|\x_{i,k}\|^2=T_i$ in \eqref{eq:constraints}, but requires $\mathbb{E}[\|\h_k\|^2]$ to be known a priori.

In this case, let $\M_{i}$ denote the $\Ka\times\Ka$ matrix solutions of the fixed point equation \cite[eq.~(22)]{kadmon18}
\begin{equation}
\M_{i} = \left(\Delta_i\I_{\Ka}+\prod_{j\neq i}^{\bigodot}\M_j \right)^{-1}\left(\prod_{j\neq i}^{\bigodot} \M_j \right) \text{\quad $i=1,\dots,p$}
\end{equation}
where $\prod_{j\neq i}^{\bigodot}$ denotes the componentwise product operation and $\Delta_i = \sigma^2T_i^{-1}\left(\prod_{j=1}^{d+1}T_j\right)^{(1-d)/(1+d)}$.
Using \cite[eq.~(13)]{kadmon18} and \cite[eq.~(80) from the supplementary material]{kadmon18}, the MSE is asymptotically equal to
\begin{equation}
\MSE_{i,k}\xrightarrow{\{T_i\}\rightarrow\infty} 1 + \e_k^H\Big( \prod_{j\neq i}^{\bigodot} \M_j + \Delta_i\I_{\Ka}\Big)^{-1}\prod_{j\neq i}^{\bigodot}\M_j\e_k  - 2\e_k^H\M_i\e_k.
\label{eq:amp_mse}
\end{equation}

%\todo[inline]{Check Kadmon and Ganguli paper which is contradictory. In particular Eq. 74 of the supplementary material seems false}

\section{LLR computation and achievable rate}
\label{sec:llr}
%We are interested in quantifying the mutual information between $\x_{i,k}$ and $\hat\z_{i,k}$ characterizing the performance of the single-user demapper \eqref{eq:single-user}.

\subsection{Approximate equivalent model}
\label{sec:equivalent_model}
In order to derive an approximation of $p(\x_{i,k}|\{\hat\z_{j,k}\},\hat\h_k)$, we first assume that the dependency of $\x_{i,k}$ on $\{\hat\z_{j,k}\}_{j\neq i}$ is negligible while the dependency on $\hat{\h}_k$ is only expressed through its norm $\|\hat{\h}_k\|^2$ as it is an indicator of user's reliability; in other words, we assume $p(\x_{i,k}|\{\hat\z_{j,k}\},\hat\h_k) \approx p(\x_{i,k}|\hat\z_{i,k},\|\hat{\h}_k\|)$.
Then, we characterize the distribution of the output $\hat\z_{i,k}$ conditioned on $\x_{i,k}$ using only the conditional mean and variance:
\begin{equation}
\left\{\begin{array}{l}\mathbb{E}[\hat\z_{i,k}|\x_{i,k}] = \alpha_{i,k}\x_{i,k}\\
\mathbb{E}[\|\hat\z_{i,k} - \alpha_{i,k}\x_{i,k}\|^2|\x_{i,k}] = \alpha_{i,k}^2\xi_{i,k}T_i
\end{array}\right.
\end{equation}
 with the constraint $\alpha_{i,k}=\left(1+\xi_{i,k}\right)^{-1/2}$ (see Lemma~\ref{prop:mse}).
On the other hand, the phase constraint $\text{Im}(\hat\z_{i,k}^H\x_{i,k})=0$ in \eqref{eq:constraints} is not practical since it involves the transmitted vector $\x_{i,k}$. 
In practice, this means that the phase ambiguity is not resolved during the low-rank tensor approximation \eqref{eq:cpd}, and is handled instead by the single-user demapper. 
In that case, the vector constellations $\mathcal{C}_i$ should be designed to be robust to a phase ambiguity, e.g. by using a non-coherent vector modulation such as the one from \cite{ngo19}.
%We will therefore consider the remaining phase indeterminacy as an unknown parameter without assuming any distribution.
We will therefore model the remaining phase indeterminacy by an unknown phase $\varphi_{i,k}$, uniformly drawn on $[0,2\pi]$.
With this random phase, we choose to approximate the distribution of $\hat{\z}_{i,k}$ given $\x_{i,k}$ by a distribution realizing the maximum entropy under the moment constraints, i.e.
\begin{equation}
\hat\z_{i,k} = e^{i\varphi_{i,k}}\alpha_{i,k}\x_{i,k} + \xi_{i,k}^{1/2}\alpha_{i,k}\w_{i,k} 
\label{eq:equivalent_channel0}
\end{equation}
where $\w_{i,k}$ is an isotropic Gaussian noise vector.
Considering eq.~\eqref{eq:receive}, we define the \textit{instantaneous signal-to-noise ratio} (SNR) for user $k$ as the signal power divided by the noise power both summed over one block of $T$ channel accesses and over the receive antennas. Taking into account the constraints $\|\x_{i,k}\|^2=T_i$, this yields $\text{SNR}_k =\frac{\|\h_k\|^2}{N\sigma^2}$.
Note that we can also define the \emph{average} SNR by replacing the term $\|\h_k\|^2$ by its expectation over the fading process.
Turning now to the point-to-point model in \eqref{eq:equivalent_channel0}, we define its \textit{equivalent SNR} as
\begin{equation}
\text{SNR}^{\text{eq}}_{i,k} = \frac{\mathbb{E}[\|e^{i\varphi_{i,k}}\alpha_{i,k}\x_{i,k}\|^2]}{\mathbb{E}[\|\hat{\z}_{i,k} - e^{i\varphi_{i,k}}\alpha_{i,k}\x_{i,k}\|^2]} = \frac{1}{\xi_{i,k}}.
\label{eq:snreq}
\end{equation}
Note that the equivalent SNR does not depend on the bias but strongly depends on $\h_k$.
Using the variance bounds established in Section~\ref{section_CRB} and the approximation \eqref{eq:mse_approx}, the point-to-point equivalent SNR can be upper bounded by
\begin{equation}
\text{SNR}^{\text{eq}}_{i,k}\leq\frac{1}{\xi_{i,k}^*} \approx  \frac{TN-T_i(\Ka-1)}{(T_i-1)}\text{SNR}_k.
%\text{SNR}^{\text{eq}}_{i,k}\leq\frac{1}{\xi_{i,k}^*} \approx \frac{TN-T_i(\Ka-1)}{(T_i-1)\sigma^2}\frac{\|\h_k\|^2}{N}= \frac{TN-T_i(\Ka-1)}{(T_i-1)}\text{SNR}_k.
\label{eq:snr_eq}
\end{equation}

\subsection{LLR of decoded symbols}
%The Log-Likelihood Ratio (LLR) is the input of many binary decoder and is defined from the conditional density of $p(\hat\z|\x_{i,k})$.
%We consider a bit mapping similar to \cite{decurninge20}
We assume in the following that the $j$-th coded bit of user $k$ (denoted by $b_{j,k}$) is among the bits mapped on $\x_{i,k}$.
In order to enable soft-decision decoding, we seek to evaluate the correponding LLR, namely 
\begin{equation}
\LLR_{j,k} = \log\frac{ p(b_{j,k}=1|\hat{\z}_{i,k})}{  p(b_{j,k}=0|\hat{\z}_{i,k})}.
\end{equation}
Denoting $\mathcal{C}_i^{(a,j)}$ the set of all symbols in $\mathcal{C}_i$ having the $j$-th bit equal to $a$ and using Bayes rule and marginalizing over $\varphi_{i,k}$, we get
\begin{eqnarray}
\LLR_{j,k} \hspace{-0.2cm}&=& \log \frac{\sum_{\x_i\in\mathcal{C}_i^{(1,j)}}  \int_0^{2\pi} p(\hat{\z}_{i,k}|\x_i,\varphi_{i,k})d\varphi_{i,k}}{\sum_{\x_i\in\mathcal{C}_i^{(0,j)}}  \int_0^{2\pi} p(\hat{\z}_{i,k}|\x_i,\varphi_{i,k})d\varphi_{i,k}}\\
 &=& \log \frac{\sum_{\x_i\in\mathcal{C}_i^{(1,j)}} \mathbb{E}_{\varphi_{i,k}}\left[ \exp\left({-\frac{\|\hat\z_{i,k} - \alpha_{i,k}e^{i\varphi_{i,k}}\x_i\|^2}{\xi_{i,k}\alpha_{i,k}^2}}\right)\right]}{\sum_{\x_i\in\mathcal{C}_i^{(0,j)}}   \mathbb{E}_{\varphi_{i,k}}\left[ \exp\left({-\frac{\|\hat\z_{i,k} - \alpha_{i,k}e^{i\varphi_{i,k}}\x_i\|^2}{\xi_{i,k}\alpha_{i,k}^2}}\right)\right]}\\
  &=& \log \frac{\sum_{\x_i\in\mathcal{C}_i^{(1,j)}} I_0\left(\frac{2}{\alpha_{i,k}\xi_{i,k}}|\hat{\z}_{i,k}^H\x_{i}|\right) }{\sum_{\x_i\in\mathcal{C}_i^{(0,j)}}   I_0\left( \frac{2}{\alpha_{i,k}\xi_{i,k}}|\hat{\z}_{i,k}^H\x_{i}|\right)}
\end{eqnarray}
where we used that $\|\x_i\|^2=T_i$ for any $\x_i\in\mathcal{C}_i$ and $I_0$ denotes the zeroth-order modified Bessel function.
Since $I_0$ is approximately exponential for large values, we have
\begin{eqnarray}
\LLR_{j,k} \approx  \frac{2}{\alpha_{i,k}\xi_{i,k}}\left(\max_{\x_i\in\mathcal{C}_i^{(1,j)}}|\hat{\z}_{i,k}^H\x_i| - \max_{\x_i\in\mathcal{C}_i^{(0,j)}} |\hat{\z}_{i,k}^H\x_i|\right).
\label{eq:llr1}
\end{eqnarray}

\begin{comment}
%================GLRT justification===============
%However, in the model  \eqref{eq:equivalent_channel0}, the density $p(b_{j,k}|\hat{\z}_{i,k})$ depends on the nuisance parameter $\varphi_{i,k}$.
%Without any distribution information on $\varphi_{i,k}$, the LLR is then ill-defined.
Instead, we consider similarly to \cite{aval03} a Generalized Likelihood Ratio Test (GLRT) metric
\begin{equation}
\GLRT_{j,k} = \log\frac{ \max_{\varphi_{i,k}\in [0,2\pi]} p(b_{j,k}=1|\hat{\z}_{i,k},\varphi_{i,k})}{ \max_{\varphi_{i,k}\in [0,2\pi]} p(b_{j,k}=0|\hat{\z}_{i,k},\varphi_{i,k})}.
\end{equation}
Using the Bayes rule and the following equalities
\begin{eqnarray}
&& \max_{\varphi_{i,k}\in [0,2\pi]}\log p(\hat\z_{i,k}|\varphi_{i,k},\x_{i,k})\label{eq:glrt}\\
&=& -\frac{1}{\alpha_{i,k}^2\xi_{i,k}}\Bigg\|\hat\z_{i,k} - \alpha_{i,k}\x_{i,k}\frac{\x_{i,k}^H\hat\z_{i,k}}{|\x_{i,k}^H\hat\z_{i,k}|}\Bigg\|^2\\
 &=& \frac{2}{\alpha_{i,k}\xi_{i,k}}|\hat{\z}_{i,k}^H\x_{i,k}|,
\end{eqnarray}
we can approximate the GLRT by
\begin{equation}
\GLRT_{j,k} \approx \log \frac{\max_{\x_i\in\mathcal{C}_i^{(1,j)}} \exp\Big(\frac{2}{\alpha_{i,k}\xi_{i,k}}|\hat{\z}_{i,k}^H\x_i|\Big)}{\max_{\x_i\in\mathcal{C}_i^{(0,j)}} \exp\Big(\frac{2}{\alpha_{i,k}\xi_{i,k}}|\hat{\z}_{i,k}^H\x_i|\Big)}
\label{eq:llr1}
\end{equation}
where $\mathcal{C}_i^{(b,j)}$ denotes the set of all possible symbols in $\mathcal{C}_i$ with bit $j$ equal to $b$.
\end{comment}

Note that this approximation is valid when either the SNR is high or $T$ is large.
Finally, in order to allow the receiver to evaluate the LLR, we approximate $\alpha_{i,k}\xi_{i,k}$ by using \eqref{eq:alpha} and the lower bound approximation from \eqref{eq:mse_approx}.
Then, \eqref{eq:llr1} depends on the channel realization only through its norm $\|\h_k\|$.
We replace this (unknown) true channel norm by its estimate $\|\hat\h_k\|^2$, i.e.
\begin{eqnarray*}
&&\frac{1}{\alpha_{i,k}\xi_{i,k}} = \frac{1}{\xi_{i,k}}\left(1+\frac{\xi_{i,k}}{T_i}\right)^{1/2}\\
&\approx & \frac{TN-T_i(\Ka-1)}{(T_i-1)\sigma^2}\frac{\|\hat\h_k\|^2}{N} \left(1+\frac{(T_i-1)\sigma^2}{TN-T_i(\Ka-1)}\frac{N}{T_i\|\hat\h_k\|^2}\right)^{1/2}.
\end{eqnarray*}

\subsection{Achievable rate using dependence testing (DT) bound}
%Let us note $\x_k=(\x_{1,k}^T,\dots,\x_{d,k}^T)^T$ and $\hat\z_k=(\hat\z_{1,k}^T,\dots,\hat\z_{d,k}^T)^T$ the concatenation of respectively the input and output modes.
Let us consider the system constituted by $d$ independent parallel channels, each following model \eqref{eq:equivalent_channel0}.
Using \cite[Th. 17]{polyanskyi10}, there exists a code with $2^B$ codewords and average probability of error not exceeding 
\begin{equation}
\epsilon \leq \mathbb{E}\left[\exp\left(-\left( i(\x_{1,k},\dots,\x_{d,k};\hat\z_{1,k},\dots,\hat\z_{d,k})-\log_2\big(\frac{2^B-1}{2}\big)\right)^+\right)\right]
\label{eq:dt}
\end{equation}
where $(\cdot)^+=\max(\cdot,0)$.
On the other hand, considering inputs $\x_{i,k}$ independent and uniformly drawn in the sphere of radius $\sqrt{T_i}$ we have (see details in Appendix~\ref{app:information})
\begin{eqnarray}
i(\x_{i,k};\hat\z_{i,k}) &=& \log_2 \frac{p(\hat\z_{i,k}|\x_{i,k})}{p(\hat\z_{i,k})}\\
 &\hspace{-3cm} =& \hspace{-1.5cm}\log_2\frac{I_0\left(\frac{2}{\alpha_{i,k}\xi_{i,k}}|\hat\z_{i,k}^H\x_{i,k}|\right) \left(\frac{\sqrt{T_i}}{\alpha_{i,k}\xi_{i,k}}\|\hat\z_{i,k}\|\right)^{T_i-1}}{2\Gamma\left(T_i+1\right)I_{T_i-1}\left(\frac{2\sqrt{T_i}}{\alpha_{i,k}\xi_{i,k}}\|\hat\z_{i,k}\|\right)}
\label{eq:information_density}
\end{eqnarray}
with $I_{n}(\cdot)$ the $n$th-order modified Bessel function.
Note that, for large $T_i$, the information density is equal to a deterministic value
\begin{eqnarray}
\frac{1}{T}i(\x_{i,k};\hat\z_{i,k}) &=& (1-\frac{1}{T})\log_2(\frac{2T}{\xi_{i,k}\alpha_{i,k}})\\
 &+& 2(\frac{1}{\xi_{i,k}}-\frac{1}{\xi_{i,k}\alpha_{i,k}})\log_2(e)-\frac{\log_2(\Gamma(T+1))}{T}\nonumber.
\end{eqnarray}

\section{Simulation results}
\label{sec:simul}
We consider TBM transmission in a SIMO Rayleigh fading scenario similar to \cite{decurninge20}, with $N=50$ receive antennas, $T=3200$ channel accesses, and equal SNR across the users. Tensor dimensions of the symbols are chosen as $(T_1,T_2) = (64,50)$ resulting in $3$-D tensor with the additional spatial dimension $N=50$.

\subsection{Comparison of MSE with bounds}
We first seek to validate eq.~\eqref{eq:mse_lb}, which provides the equivalent noise level subsequently used in the equivalent single-user channel model \eqref{eq:equivalent_channel0}.
The MSE performance resulting from solving \eqref{eq:cpd} under constraints \eqref{eq:constraints} was evaluated through a Monte-Carlo simulation with $\x_{i,k}$ drawn from a uniform distribution on the sphere of radius $\sqrt{T_i}$. %, $\h_k$. and $\w$ drawn as Gaussian random vectors, hence the transmitters are assumed to have the same SNR.
An approximate solution was computed using an inexact Gauss-Newton (GN) algorithm \cite{sorber13} with a maximum of 100 iterations.
Figs.~\ref{fig:snr_hist} and \ref{fig:snr_hist100} depict the result for $\Ka=1$ and 100 respectively, as well as the approximate lower bound \eqref{eq:mse_lb} and the theoretical asymptotic performance of AMP from \cite{kadmon18} (we consider mode 1 of user 1 w.l.o.g).
Since it is not possible to ensure that the non-convex problem \eqref{eq:cpd} is solved optimally by the GN algorithm, we resort to genie-initialized GN to try and avoid local minima, thus approaching the performance of an optimal solution to \eqref{eq:cpd}.
The genie initialization consists in initializing the gradient descent with the transmitted symbols and channels.

%Note that the phase constraint in \eqref{eq:constraints} requires the knowledge of the transmitted vectors $\x_{i,k}$ and that these two figures highlight the performances of the tensor decomposition only without taking the constellation into account.
%Two examples of MSE for the estimator solution of \eqref{eq:cpd} are depicted in Figure~\ref{fig:snr_hist} for $\Ka=1$ respectively, and compared to the approximated lower bound.
The histograms depicted at the top of Figure~\ref{fig:snr_hist} and Figure~\ref{fig:snr_hist100} represent the distribution of MSE values for selected SNR values (the mean of the histogram corresponds to the Y-axis position on the bottom part of the figure).
%We see in particular that the number of users has a small influence on the equivalent SNR as soon as $\frac{TN}{T_i}\gg \Ka$. 

At low SNR, we observe that the bound \eqref{eq:mse_approx} is not tight.
This suggests that this bound cannot be achieved in this regime.
In order to explain this phenomenon, Figure~\ref{fig:snr_hist} details the histograms for particular low SNR values. 
We see that, in this regime, the histograms of MSE over the runs is bimodal.
%Since it is not possible to ensure that \eqref{eq:cpd} is solved optimally in our numerical experiments, we resort to genie initialization to avoid the problem of local minima and approximate the performance of an optimal solution to the low-rank tensor decomposition.
Considering the curve corresponding to a genie initialization, we observe that the mode of the top in the histograms corresponds either to a global minimum of \eqref{eq:cpd} far for the transmitted symbol or to the convergence to a local minimum due to poor initialization.
Furthermore, we observe that the theoretical AMP performance is close to that of GN with genie initialization.
In particular, the phase transition of the AMP analysis around $-29$ dB is clearly visible while it is not present in the lower bound \eqref{eq:mse_approx}.
Finally, note that we observe in Figure~\ref{fig:snr_hist100} a gap between the performance of the GN with random initialization and \eqref{eq:mse_approx} due to the limited number of iterations and the high accuracy in the convergence asked in order to achieve a MSE of $-40$ dB.

\begin{figure}[h!]
\centering
\tiny
\input{mse_T3200_hist1.tikz}
\caption{$\MSE_{1,1}$ vs. SNR  for $\Ka=1$ user.}
\label{fig:snr_hist}
\end{figure}
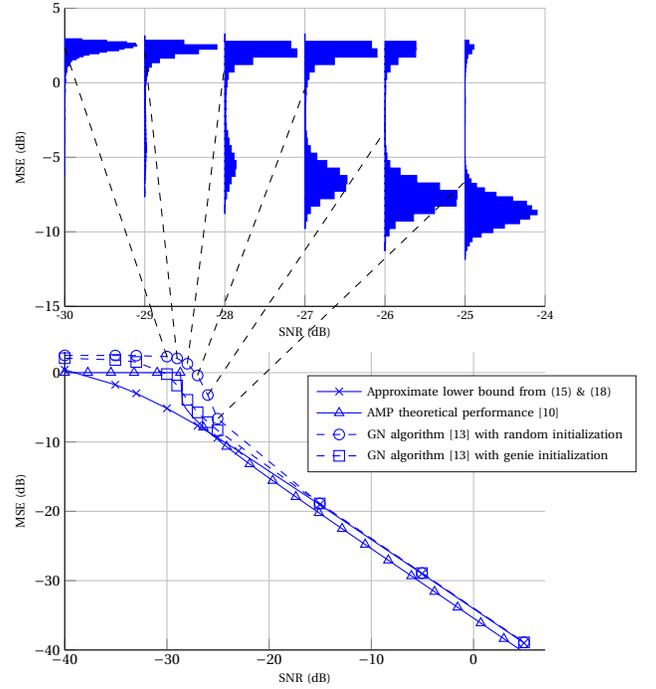

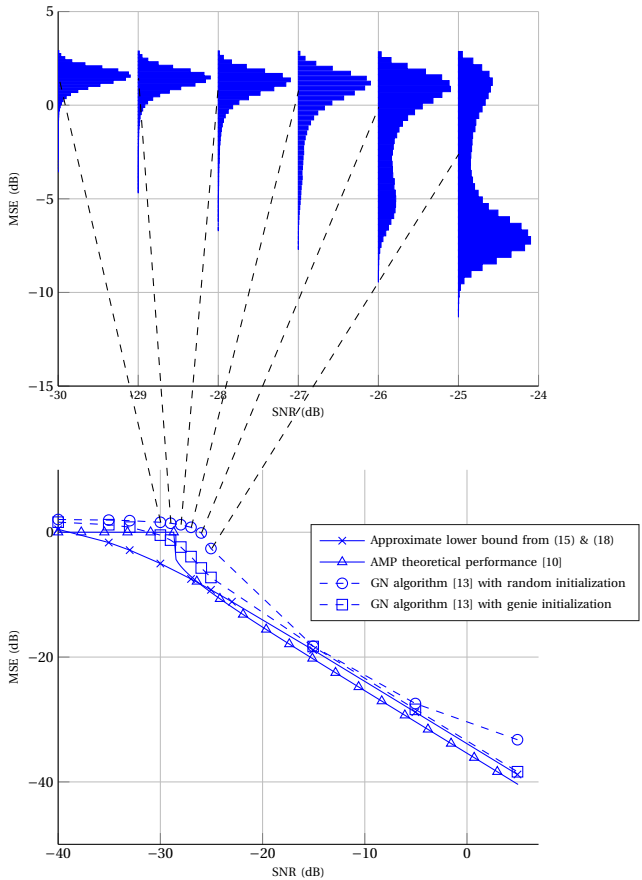
\begin{figure}
\scriptsize
\centering
\tiny
\input{mse_T3200_hist100.tikz}
\caption{$\MSE_{1,1}$ vs. SNR  for $\Ka=100$ users.}
\label{fig:snr_hist100}
\end{figure}

\subsection{Packet Error Rate with Channel Code}
%We assume that the receiver knows the number of active users $\Ka$ and fix the payload to $300$ bits.
In Figure~\ref{fig:bler_snr}, we evaluate the performance of TBM in terms of Packet Error Rate (PER) for a payload of $300$ bits encoded with either a BCH code with $318$ coded bits and hard-decision decoding (as in \cite{decurninge20}), or the rate $0.85$ polar code from the 5G standard under soft-decision decoding, using as sub-constellation $\mathcal{C}_i$ either (i) the codebook design from \cite{ngo19} or (ii) a vector constellation composed with one coordinate used as pilot and QAM symbols.
For comparison, we also depict the performance of the single-user constellation directly simulating $d$ parallel instances of model \eqref{eq:equivalent_channel0} skipping the tensor decomposition step.
Comparing the latter with the performance of TBM scheme shows \eqref{eq:equivalent_channel0} is a good approximation of the tensor decomposition output.
Finally we also represent the PER corresponding to the DT bound \eqref{eq:dt} where the expectation is estimated by generating $1000$ random realizations of $(\hat\z_{i,k},\x_{i,k})$ generated using the model \eqref{eq:equivalent_channel0}.
The gap between the DT bound and the practical constellations used in \eqref{eq:equivalent_channel0} can be explained by the sub-optimalities of the binary code and the chosen vector constellation and the randomness of the channel gains $\|\h_k\|^2$ which are not taken into account in the DT bound.
Additionally, we see in all cases the small influence of the number of active users $\Ka$ on the PER both theoretically and numerically.
Note however that the approximations used in the paper assumed $\prod_{j\neq i}T_j \gg \Ka$.
We oberve in practice that, for larger values of $\Ka$ the tensor decomposition fails and the PER is equal to 1 for any $\SNR$ value (see \cite{decurninge20}).

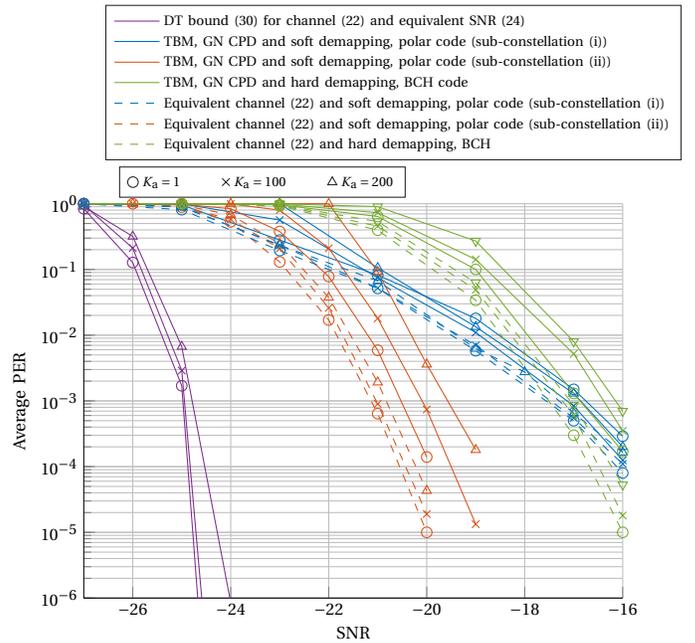
\begin{figure}[h!]
\scriptsize
\centering
\input{blervssnr300_full.tikz}
\caption{PER vs. SNR with $300$ bits payload.}
\label{fig:bler_snr}
\end{figure}

%\begin{figure}[h!]
%\scriptsize
%\centering
%\input{figs/blervsKa300.tikz}
%\caption{PER vs. $\Ka$ with $300$ bits payload at $\SNR=-19$dB.}
%\label{fig:bler_Ka}
%\end{figure}

\section{Conclusion}
We characterized the accuracy of tensor decomposition in the receiver of a TBM-based system.
We introduced an equivalent single-user channel model allowing to derive approximations of the LLR. 
The PER was evaluated for the equivalent model, and compared to the DT achievable bound.
These results were shown to be in good agreement with the performance achieved by TBM with BCH and polar codes through numerical simulations.

\bibliographystyle{IEEEtran}
\bibliography{refs}

\appendices

\section{Proof Of Proposition~\ref{prop:crb}}
\label{app:crb}
In order to derive a lower bound of the constrained estimators $\hat\btheta$, we will use the Cram\'er-Rao bounds.
Let us define $\mathcal{I}_{\btheta\btheta}$ the complex Fisher information matrix equal to $\mathcal{I}_{\btheta\btheta} = \mathbb{E}\Big[\frac{\partial \ell}{\partial \btheta^*} \Big(\frac{\partial \ell}{\partial \btheta^*}\Big)^H \Big]$ where $\btheta^*$ denotes the complex conjugate of the vector $\btheta$..

Note that since $\btheta$ is a complex random vector,  the ``full'' information matrix is the information matrix of the real vector $\btheta_R=(\text{Re}(\btheta),\text{Im}(\btheta))^T$. Moreover, since the noise $\w$ is circularly symmetric, the full information matrix is equal to (see \cite{liu_sidiropoulos_01})
\begin{equation}
\mathcal{I}_{RR} = \mathbb{E}\Big[\frac{\partial \ell}{\partial \btheta_R} \Big(\frac{\partial \ell}{\partial \btheta_R}\Big)^T \Big] = \mathcal{M}_{K\sum_i T_i}\left(\begin{array}{cc} \mathcal{I}_{\btheta\btheta} & 0\\ 0 & \mathcal{I}_{\btheta\btheta}^*\end{array}\right)\mathcal{M}_{K\sum_i T_i}^H
\end{equation}
with
\begin{equation}
\mathcal{M}_{K\sum_i T_i}=\frac{1}{\sqrt{2}}\left(\begin{array}{cc} \I_{K\sum_i T_i} & \I_{K\sum_i T_i}\\ -i\I_{K\sum_i T_i} & i\I_{K\sum_i T_i}\end{array}\right)
\end{equation}
therefore, is completely characterized by $\mathcal{I}_{\btheta\btheta}$.

In order to derive the Cram\'er-Rao bound for the constrained parameters \eqref{eq:constraints}, we note $\C(\hat\btheta_R) = (C_{1,1,1}(\hat\btheta),\dots,C_{p-1,K,2}(\hat\btheta))^T$ with the $K(p-1)$ constraint functions $C_{i,k,1}(\hat\btheta) = \|\hat\z_{i,k}\|^2-T_i$ and $C_{i,k,2}(\hat\btheta)= \text{Im}(\x_{i,k}^H\hat\z_{i,k})$ for $1\leq i\leq p-1$.
Furthermore, we define $\nu_{\btheta_R}$ an orthonormal basis of the subspace of $\mathbb{R}^{2K\sum_{i=1}^pT_i}$ orthogonal to the columns of the Jacobian of $\C(\btheta_R)$.
Let $\bdiag(\cdot)$ denote the block diagonal operator and $\D_{\alpha} = \bdiag(\alpha_{1,1}^2\I_{T_1},\alpha_{1,2}^2\I_{T_1},\dots,\alpha_{p,\Ka}^2\I_{T_p}$.
Then \eqref{eq:bias} rewrites
\[
\mathbb{E}[\hat\btheta_R] = \bdiag(\D_{\alpha},\D_{\alpha})\hat\btheta_R 
\]
and a lower bound on the constrained estimator is given by (see \cite{benhaim09})
\begin{eqnarray}
&&\mathbb{E}[(\hat\btheta_R - \mathbb{E}[\hat\btheta_R])(\hat\btheta_R - \mathbb{E}[\hat\btheta_R])^H]\nonumber \\
 &\succeq& \bdiag(\D_{\alpha},\D_{\alpha})\nu_{\btheta_R}(\nu_{\btheta_R}^H\mathcal{I}_{RR}\nu_{\btheta_R})^{-1}\nu_{\btheta_R}^H.
\label{eq:cramerrao_real}
\end{eqnarray}

Therefore, since $\left.\frac{dC_{i,k,1}}{d\hat\z_{i,k,R}}\right|_{\hat\z_{i,k}=\x_{i,k}} = \left(\begin{array}{cc}\text{Re}(\x_{i,k})\\ \text{Im}(\x_{i,k})\end{array}\right)$ and  $\left.\frac{dC_{i,k,2}}{d\hat\z_{i,k,R}}\right|_{\hat\z_{i,k}=\x_{i,k}} =  \left(\begin{array}{c} -\text{Im}(\x_{i,k}) \\ \text{Re}(\x_{i,k}) \end{array}\right)$ with the notation $\hat\z_{i,k,R}=\left(\begin{array}{cc}\text{Re}(\hat\z_{i,k})\\ \text{Im}(\hat\z_{i,k})\end{array}\right)$, we get
\begin{eqnarray}
\nu_{\btheta_R} &=& \left(\begin{array}{cc}\text{Re}(\nu_{\btheta}) & -\text{Im}(\nu_{\btheta})\\ \text{Im}(\nu_{\btheta}) & \text{Re}(\nu_{\btheta})\end{array}\right)\\
&=& \mathcal{M}_{K\sum_i T_i}\left(\begin{array}{cc}\nu_{\btheta} & 0\\ 0 & \nu_{\btheta}^*\end{array}\right)\mathcal{M}_{K(\sum_i T_i-p+1)}^H
\end{eqnarray}
with $\nu_{\btheta}=\bdiag(P_{\x_{1,1}}^{\perp},P_{\x_{2,1}}^{\perp},\dots,P_{\x_{p-1,k}}^{\perp},\I_{T_P},\dots,\I_{T_P})$ with the notation $P_{\x_{i,k}}^{\perp}\in\mathbb{C}^{T_i\times (T_i-1)}$ expressing a basis of vectors orthogonal to $\x_{i,k}$.

The Cram\'er-Rao bound can then be expressed in the complex domain as 
\begin{equation}
\mathbb{E}\Big[\big(\hat{\btheta} - \btheta\big)\big(\hat{\btheta} - \btheta\big)^T\Big]\succeq \D_{\alpha}\nu_{\btheta}(\nu_{\btheta}^H\mathcal{I}_{\btheta\btheta}\nu_{\btheta})^{-1}\nu_{\btheta}^H.
\label{eq:cr_cplx}
\end{equation}
Let us now explicit the matrix $\mathcal{I}_{\btheta\btheta}$.
Let us note $\s_{-ij,k}$ the symbol sent by user $k$ where we remove mode $i$ and $j$, i.e. for $1\leq i,j\leq p$
\begin{equation}
\s_{-ij,k} = \x_{1,k}\otimes \dots\otimes \x_{i-1,k}\otimes \x_{i+1,k}\otimes\dots\otimes \x_{j-1,k}\otimes \x_{j+1,k}\otimes\dots\otimes\x_{p,k}
\end{equation}
and
\begin{equation}
\left\{\begin{array}{lll}
\S_{-ij} &=& (\s_{-ij,1},\dots,\s_{-ij,K})\in\mathbb{C}^{\prod_{m\neq i,j}T_m\times K},\\
\bGamma_{ij} &=& [\gamma_{ij,kk'}]_{k,k'=1}^K= \S_{-ij}^H\S_{-ij},\\
\G &=& \bdiag(\bGamma_{ii}\otimes \I_{T_i})_{i=1}^p\in\mathbb{C}^{K\sum_iT_i\times K\sum_iT_i},\\
 Z &=& \bdiag(\I_K\otimes \X_i)_{i=1}^p\in\mathbb{C}^{K\sum_iT_i\times K^2p}.
\end{array}\right.
\end{equation}
By noting $\P\in\mathbb{C}^{K^2\times K^2}$ the permutation matrix such that $\vec(\M) = \P\vec(\M^T)$ for any matrix $\M$, we define $\K$ by
\begin{equation}
\K = [\K_{ij}]_{i,j=1}^p \text{\quad with \quad} \K_{ij} = (1-\delta_{ij})\P\diag(\vec(\bGamma_{ij})) 
\end{equation}
with $\delta_{ij}=\left\{\begin{array}{cc}1\text{ if $i=j$}\\0\text{ if $i\neq j$}\end{array}\right.$.
With these notations, we get from \cite{tichavsky13}
\begin{equation}
\mathcal{I}_{\btheta\btheta} = \frac{1}{\sigma^2}(\G + \Z\K\Z^H).
\end{equation}
Then, noting $\G_{\nu} = \nu_{\btheta}^H\G\nu_{\btheta}$ and $\Z_{\nu} =  \nu_{\btheta}^H\Z$. and using Woodbury identity, we get 
\begin{eqnarray}
&&(\nu_{\btheta}^H\mathcal{I}_{\btheta\btheta}\nu_{\btheta})^{-1} \\
&=&\sigma^2\G_{\nu}^{-1} - \sigma^2\G_{\nu}^{-1}\Z_{\nu}\K(\I+\Z_{\nu}^H\G_{\nu}^{-1}\Z_{\nu}\K)^{-1}\Z_{\nu}^H\G_{\nu}^{-1}\\
&=&\sigma^2\G_{\nu}^{-1/2}(\I - \M + \M(\I+\M)^{-1}\M)\G_{\nu}^{-1/2}.
\end{eqnarray}
where we noted $\M = \G_{\nu}^{-1/2}\Z_{\nu}\K\Z_{\nu}^H\G_{\nu}^{-1/2}$.
Since $\nu_{\btheta}^H\mathcal{I}_{\btheta\btheta}\nu_{\btheta} = (\G_{\nu} + \Z_{\nu}\K\Z_{\nu}^H)\succ 0$, it holds that $(\I+\M)^{-1}$ is a positive matrix, we get a lower bound on the inverse information matrix, i.e.
\begin{equation}
(\nu_{\btheta}^H\mathcal{I}_{\btheta\btheta}\nu_{\btheta})^{-1} \succeq 
\sigma^2\G_{\nu}^{-1} - \sigma^2\G_{\nu}^{-1}\Z_{\nu}\K\Z_{\nu}^H\G_{\nu}^{-1}.
\label{eq:fisher_lowerbound}
\end{equation}
Since $\G_{\nu}$ and $\Z_{\nu}$ are block-diagonal matrices with $p$ blocks and the block diagonal elements of $\K$ are equal to 0, it results that the diagonal elements of $\G_{\nu}^{-1}\Z_{\nu}\K\Z_{\nu}^H\G_{\nu}^{-1}$ are equal to 0. 

Note $\E_{i,k}\in\mathbb{C}^{K\sum_jT_j \times T_i}$ the projector on the $i$-th mode of the $k$-user and $\E_{k}\in\mathbb{C}^{KT_i\times T_i}$ the projector on user $k$.
Since the matrices $\G$ and $\nu_{\btheta}$ are block-diagonal with respect to the index $1\leq i\leq p$, we can only consider the i-th block of the matrices $\G$ and $\nu_{\btheta}$ that we denote by $\G_i$ and $\nu_i$ respectively.
Then, for any $1\leq i<p$, \eqref{eq:fisher_lowerbound} yields
\begin{equation}
\E_{i,k}^H\nu_{\btheta}(\nu_{\btheta}^H\mathcal{I}_{\btheta\btheta}\nu_{\btheta})^{-1}\nu_{\btheta}^H\E_{i,k} \succeq \sigma^2 \E_{k}^H\nu_{i}(\nu_i^H\G_i\nu_i)^{-1}\nu_{i}^H\E_{k}.
\label{eq:cr_inequality}
\end{equation}
On the other hand, if we complete the incomplete basis $\nu_{i}$ with a collection of orthogonal vectors stacked in $\mu_{i}$, we get
\begin{equation}
(\nu_{i}^H\G_i\nu_{i})^{-1} = \nu_{i}^H\G_i^{-1}\nu_{i} - \nu_{i}^H\G_i^{-1} \mu_{i}(\mu_{i}^H\G_i^{-1}\mu_{i})^{-1}\mu_{i}^H\G_i^{-1}\nu_{i}.
\label{eq:decomposition}
\end{equation}
Let us denote an orthonormal basis of the subspace of $\mathbb{C}^{T_iK}$ spanned by the columns of $\G_i^{-1/2}\mu_{i}$ by $\overline\U_i$.
On the other hand, since $\mu_i = \bdiag(\frac{\x_{i,1}}{\|\x_{i,1}\|},\dots,\frac{\x_{i,K}}{\|\x_{i,K}\|})$, we obtain that $\G_i^{-1/2}\mu_{i}=(\bGamma_{ii}^{-1/2}\e_{k'}\otimes \frac{\x_{i,k'}}{\|\x_{i,k'}\|})_{1\leq k'\leq K}$.
Then, it holds
\begin{equation}
\G_i^{-1/2} \mu_{i}(\mu_{i}^H\G_i^{-1}\mu_{i})^{-1}\mu_{i}^H\G_i^{-1/2} = \overline\U_i\overline\U_i^H.
\label{eq:Ui}
\end{equation}
Moreover, if we note $\blambda_{i,k} = \bGamma_{i,i}^{-1/2}\e_k$ and $\V_i = \frac{\blambda_{i,k}}{\|\blambda_{i,k}\|}\otimes P^{\perp}_{\x_{i,k}}$ , we get
\begin{equation}
\G_i^{-1/2}\nu_{i}\nu_{i}^H\E_{k} = \blambda_{i,k}\otimes P^{\perp}_{\x_{i,k}}(P^{\perp}_{\x_{i,k}})^H
\end{equation}
and
\begin{eqnarray}
&&\G_i^{-1/2}\nu_{i}\nu_{i}^H\E_{k}\E_{k}^H\nu_{i}\nu_{i}^H\G_i^{-1/2}\nonumber\\
 &=&  \blambda_{i,k}\blambda_{i,k}^H\otimes P^{\perp}_{\x_{i,k}}(P^{\perp}_{\x_{i,k}})^H\nonumber\\
 &=& \| \blambda_{i,k}\|^2\V_i\V_i^H = \e_k^H\bGamma_{i,i}^{-1}\e_k\V_i\V_i^H\nonumber\\
 &=&  (\bGamma_{i,i}^{-1})_{k,k}\V_i\V_i^H\label{eq:Vi}.
\end{eqnarray}

Hence, combining \eqref{eq:cr_cplx}, \eqref{eq:cr_inequality}, \eqref{eq:Ui} and \eqref{eq:Vi} the variance of any estimator satisfying \eqref{eq:bias} is lower bounded by 
\begin{eqnarray}
&&\frac{1}{\sigma^2\alpha_{i,k}^2}\mathbb{E}\big[\|\hat{\z}_{i,k} -\alpha_{i,k} \x_{i,k}\|^2\big]\\
&\geq&\text{Tr}\Big[\E_{i,k}^H\nu_{\btheta}(\nu_{\btheta}^H\mathcal{I}_{\btheta\btheta}\nu_{\btheta})^{-1}\nu_{\btheta}^H\E_{i,k}\Big]\label{eq:cr00}\\
%&\geq&  \text{Tr}\Big[\E_{k}^H\nu_{i}\nu_{i}^H\G_i^{-1}\nu_{i}\nu_{i}^H\E_{k}\Big] - \text{Tr}\Big[\E_{k}^H\nu_{i}\nu_{i}^H\G_i^{-1/2}\U_i\U_i^H\G_i^{-1/2}\nu_{i}\nu_{i}^H\E_{k}\Big]\\
&\geq& \text{Tr}\Big[P^{\perp}_{\x_{i,k}}(P^{\perp}_{\x_{i,k}})^H (\bGamma_{i,i}^{-1})_{k,k} -  (\bGamma_{i,i}^{-1})_{k,k}\V_i\V_i^H \overline\U_i\overline\U_i^H\Big]\\
&=& \frac{T_i-1-\|\V_i^H \overline\U_i\|^2}{\|P^{\perp}_{\S_{-ii,-k}}\s_{-ii,k}\|^2}
\end{eqnarray}
By remarking that the vector $\bGamma_{ii}^{-1/2}\e_k\otimes \x_{i,k}$ is orthogonal to $\V_i$, we can remove its corresponding subspace in $\overline\U_i$ without changing the result.

Finally, in the case $K=1$, we get $\|\overline\U_i^H\V_i\|^2=0$ and $\bGamma_{ii} = \prod_{j\neq i}\|\x_{i,1}\|^2$ leading to the result in \eqref{eq:cr_variance1}.

\section{Proof Of Proposition~\ref{prop:crb2}}
\label{app:crb2}
Note $\D_i=\diag((\bGamma_{ii}^{-1})_{11},\dots,(\bGamma_{ii}^{-1})_{KK})$.
Since $\bGamma_{ii}$ is invertible, we get that $(\bGamma_{ii}^{-1})_{kk} = \e_k^H\bGamma_{ii}^{-1}\e_k>0$ for any $1\leq k\leq K$, i.e. there exists $\eta_i^-,\eta_i^+$ such that $\eta_i^- \D_i \preceq \bGamma_{ii}^{-1}\preceq \eta_i^+ \D_i$.
Combining this property with \eqref{eq:cr_inequality} and \eqref{eq:decomposition}, we will now derive an alternative lower bound on the estimator variance $\mathbb{E}\big[\|\hat{\z}_{i,k} - \alpha_{i,k}\x_{i,k}\|^2\big]$.

Note first that it holds that $\eta_i^-\D_i\otimes\I_{T_i} \preceq \G_i^{-1} \preceq \eta_i^+\D_i\otimes\I_{T_i}$ and in particular 
\[
(\mu_i^H\G_i^{-1}\mu_i)^{-1} \preceq \frac{1}{\eta_i^-}\D_i^{-1}.
\]
Then, considering the $(k,k)$-th block of $\nu_{i}^H\G_i^{-1} \mu_{i}(\mu_{i}^H\G_i^{-1}\mu_{i})^{-1}\mu_{i}^H\G_i^{-1}\nu_{i}$ and noting $\tilde\x_{i,k} = \frac{\x_{i,k}}{\|\x_{i,k}\|}$, it holds
\begin{eqnarray}
&&\E_k^H\nu_{i}^H\G_i^{-1} \mu_{i}(\mu_{i}^H\G_i^{-1}\mu_{i})^{-1}\mu_{i}^H\G_i^{-1}\nu_{i}\E_k\\
 &\preceq& \frac{1}{\eta_i^-}(P^{\perp}_{\x_{i,k}})^H\left(\sum_{k'=1}^K \big|(\bGamma_{ii}^{-1})_{kk'}\big|^2 \frac{1}{(\bGamma_{ii}^{-1})_{k'k'}}\tilde\x_{i,k'}\tilde\x_{i,k'}^H\right)P^{\perp}_{\x_{i,k}}.
\end{eqnarray}
Considering the trace and using the fact that $P^{\perp}_{\x_{i,k}}\tilde\x_{i,k}=0$ and $\|P^{\perp}_{\x_{i,k}}\tilde\x_{i,k'}\|^2\leq 1$ for any $k\neq k'$, we get
\begin{eqnarray}
&&\text{Tr}\big[\E_k^H\nu_{i}^H\G_i^{-1} \mu_{i}(\mu_{i}^H\G_i^{-1}\mu_{i})^{-1}\mu_{i}^H\G_i^{-1}\nu_{i}\E_k\big]\\
&\leq&  \frac{1}{\eta_i^-}\sum_{k'\neq k} \big|(\bGamma_{ii}^{-1})_{kk'}\big|^2 \frac{1}{(\bGamma_{ii}^{-1})_{k'k'}}\big\|(P^{\perp}_{\x_{i,k}})^H\tilde\x_{k'}\big\|^2\\
&\leq&  \frac{1}{\eta_i^-}\sum_{k'\neq k} \big|(\bGamma_{ii}^{-1})_{kk'}\big|^2 \frac{1}{(\bGamma_{ii}^{-1})_{k'k'}}\\
&\leq& \frac{1}{\eta_i^-}\left( (\bGamma_{ii}^{-1}\D_i^{-1}\bGamma_{ii}^{-1})_{kk} - \frac{\big|(\bGamma_{ii}^{-1})_{kk}\big|^2}{(\D_i)_{kk}}\right)\label{eq:upper}.
\end{eqnarray}
Since $\D_i^{-1/2}\bGamma_{ii}^{-1}\D_i^{-1/2} \preceq \eta_i^+ \I_K$, we get $\bGamma_{ii}^{-1}\D_i^{-1}\bGamma_{ii}^{-1}\preceq (\eta_i^+)^2\D_i$.
Moreover, it holds $\frac{\big|(\bGamma_{ii}^{-1})_{kk}\big|^2}{(\D_i)_{kk}}=(\bGamma_{ii}^{-1})_{kk}$.

Therefore, combining \eqref{eq:decomposition} and \eqref{eq:upper}
\begin{eqnarray}
&&\frac{1}{\sigma^2\alpha_{i,k}^2}\mathbb{E}\big[\|\hat{\z}_{i,k} - \alpha_{i,k}\x_{i,k}\|^2\big]\\
 &\leq& (T_i-1)(\bGamma_{ii}^{-1})_{kk} - \frac{(\bGamma_{ii}^{-1})_{kk}}{\eta_i^-}((\eta_i^+)^2-1)\\
 &=& \frac{T_i-1}{\|P^{\perp}_{\S_{-ii,-k}}\s_{-ii,k}\|^2}\left(1 - \frac{(\eta_i^+)^2-1}{(T_i-1)\eta_i^-}\right).
\end{eqnarray}

\section{Proof of Lemma~\ref{prop:mse}}
\label{app:bias_variance}
First note that in the class of estimators satisfying the constraints \eqref{eq:constraints}, it holds that $\mathbb{E}[\|\hat\z_{i,k}\|^2]=\|\x_{i,k}\|^2=T_i$, hence, using \eqref{eq:bias}, the bias and the variance of $\hat\z_{i,k}$ are related for any $1\leq i\leq p-1$ through
\begin{eqnarray}
&&\mathbb{E}[\|\hat\z_{i,k} - \alpha_{i,k}\x_{i,k}\|^2]\nonumber\\
 &=& \mathbb{E}[\|\hat\z_{i,k}\|^2] + \alpha_{i,k}^2\|\x_{i,k}\|^2 - 2\alpha_{i,k}\mathbb{E}[\text{Re}(\hat\z_{i,k}^H\x_{i,k})]\nonumber\\ 
 &=& T_i(1-\alpha_{i,k}^2).
\label{eq:bias_var}
\end{eqnarray}

Combining \eqref{eq:bias_var} and the definition of $\xi_{i,k}$ in \eqref{eq:xidef} yields that $\alpha_{i,k}$ and $\xi_{i,k}$ are linked through 
\begin{equation}
\alpha_{i,k}^2 = \frac{1}{1+ \xi_{i,k}}.
\label{eq:alpha}
\end{equation}
Repeating the operation in \eqref{eq:bias_var} for the MSE yields
\begin{eqnarray}
\text{MSE}_{i,k} &=& \frac{1}{T_i}\left( \mathbb{E}[\|\hat\z_{i,k}\|^2] +\|\x_{i,k}\|^2 - 2\mathbb{E}[\text{Re}(\hat\z_{i,k}^H\x_{i,k})]\right)\\
&=& 2(1-\alpha_{i,k})=  2\left(1-\frac{1}{(1+\xi_{i,k})^{1/2}}\right). %\left(T_i^{-1}+\xi_{i,k}^{-1}\right)^{-1}.
\label{eq:cr_mse0}
\end{eqnarray}

We can then lower bound the MSE using the definition of $\xi^*_{i,k}$ 
\begin{eqnarray}
\label{eq:mse_norm_bound}
&&\frac{1}{T_i}\mathbb{E}\big[\|\hat{\z}_{i,k} - \x_{i,k}\|^2\big] = 2\left(1 - \frac{1}{(1+\xi_{i,k})^{1/2}}\right)\\
 &\leq& 2\left(1 - \frac{1}{(1+\xi_{i,k}^*)^{1/2}}\right)
 \end{eqnarray}

\section{Computation of the the information density for model \eqref{eq:equivalent_channel0}}
\label{app:information}
For the sake of compactness of notations, we will drop the indices $i$ and $k$ and consider the model \eqref{eq:equivalent_channel0} with input $\x=\sqrt{T}\u$ with $\u$ uniformly drawn on the sphere of radius $1$ and the output $\hat\z$ so that
\begin{equation}
\hat\z = \alpha(e^{i\varphi}\sqrt{T}\u+\xi^{1/2}\w)
\end{equation}
with $\varphi$ uniformly drawn in $[0;2\pi]$ and $\w$ standard Gaussian random vector.

In the following , we will note $\beta=\frac{1}{\alpha\xi}$ and $S_{\mathbb{C},T}$ and $S_{\mathbb{R},T}$ the complex and real unit sphere respectively.

Remarking that $\u$ and $e^{i\varphi}\u$ have the same distribution and using the conditioning over $\varphi$ and $\u$, we get
\begin{eqnarray}
p(\hat\z|\x)
&=& \frac{1}{2\pi}\int_{0}^{2\pi} \exp(-\frac{1}{\xi\alpha^2}\|\hat\z-\alpha e^{i\varphi}\x\|^2)d\varphi \label{eq:line2}\\
&=& \frac{1}{2\pi}\int_{0}^{2\pi} \exp\left\{2\beta \text{Re}(\hat\z^H\x e^{i\varphi})\right\} d\varphi \\
&=& \frac{1}{2\pi}\int_{S_{\mathbb{R},2}} \exp\left\{2\beta |\hat\z^H\x|\u^T\e_1\right\} d\u\label{eq:line4}.
\end{eqnarray}
where we used in \eqref{eq:line2} that $\|\x\|^2=T$ is a constant and, in \eqref{eq:line4}, we noted $\e_1=(1 0)^T$ and used that $\text{Re}(\hat\z^H\x e^{i\varphi})$ and $|\hat\z^H\x|\text{Re}( e^{i\varphi})$ have the same distribution.
On the other hand,
\begin{eqnarray}
p(\hat\z)&=& \frac{\Gamma(T+1)}{\pi^T}\int_{S_{\mathbb{C},T}} \exp(-\frac{1}{\xi\alpha^2}\|\hat\z-\alpha\sqrt{T}\u\|^2)d\u \label{eq:line2b}\\
&=& \frac{\Gamma(T+1)}{\pi^T}\int_{S_{\mathbb{C},T}} \exp\left\{2\beta\sqrt{T}\text{Re}(\hat\z^H\u)\right\}d\u \\
&=& \frac{\Gamma(T+1)}{\pi^T}\int_{S_{\mathbb{R},2T}} \exp\left\{2\beta\sqrt{T}(\hat{\tilde\z}^T\tilde\u)\right\}d\tilde\u \label{eq:line4b}.
\end{eqnarray}
where we used in \eqref{eq:line2b} that the volume of the complex $T$-dimensional unit sphere is equal to $\frac{\pi^T}{\Gamma(T+1)}$ and, in \eqref{eq:line4}, we noted $\tilde\u = (\text{Re}(\u^T),\text{Im}(\u^T))^T$.
The information density then writes 
\begin{eqnarray}
i(\x;\hat\z) &=& \log_2\frac{p(\hat\z|\x)}{p(\hat\z)}\\
&=& \log_2\frac{\frac{1}{2\pi}\int_{S_{\mathbb{R},2}} \exp\left\{2\beta |\hat\z^H\x|\u^T\e_1\right\} d\u}{\frac{\Gamma(T+1)}{\pi^T}\int_{S_{\mathbb{R},2T}} \exp\left\{2\beta\sqrt{T}(\hat{\tilde\z}^T\tilde\u)\right\}d\tilde\u }.
\end{eqnarray}
Now, we use that for any $n\geq 1$, $\kappa>0$ and $\x\in S_{\mathbb{R},n}$, (see e.g. \cite{watson84})
\begin{equation}
\int_{S_{\mathbb{R},n}} \exp\{\kappa\u^T\x\}d\u = \frac{(2\pi)^{n/2}I_{n/2-1}(\kappa)}{\kappa^{n/2-1}}.
\end{equation}
Then, we get
\begin{equation}
i(\x;\hat\z) = \log_2\frac{I_0(2\beta |\hat\z^H\x|)}{\frac{\Gamma(T+1)}{\pi^T}\frac{(2\pi)^{T}I_{T-1}(2\beta\sqrt{T}\|\hat\z\|)}{(2\beta\sqrt{T}\|\hat\z\|)^{T-1}} }\label{eq:line4}.
\end{equation}
Now, note that 
\[
\|\hat\z\|^2 = \alpha^2(T+\xi\|\w\|^2+2(T\xi)^{1/2}\text{Re}(e^{i\varphi}\u^H\w))
\]
and 
\begin{eqnarray*}
|\hat\z^H\x|^2 &=& \alpha^2\big|Te^{i\varphi}+(T\xi)^{1/2}\u^H\w\big|^2\\
 &=&  \alpha^2((T+(T\xi)^{1/2}\text{Re}(e^{-i\varphi}\u^H\w))^2+\text{Im}(e^{-i\varphi}\u^H\w)^2)
\end{eqnarray*}
First remark that $\u$ and $e^{i\varphi}\u$ have the same distribution, and that $\|\w\|^2$ and $\frac{\w}{\|\w\|}$ are independent and respectively drawn from a  chi-square random variable with $2T$ degrees of freedom and a uniform distribution on the complex unit sphere.
We can then define four independent scalar random variables: $Q$ a chi-square random variable with $2T$ degrees of freedom and $Z_1,Z_2$ two independent random variables drawn from $\text{Beta}(1/2,(T-1)/2)$ and $S$ a Rademacher random variable, we have that 
\begin{equation}
\begin{array}{lll}
\|\hat\z\|^2 &\stackrel{\mathcal{L}}{=}& \alpha^2(T+\xi Q/2+2(T\xi QZ_1/2)^{1/2}S),\\
|\hat\z^H\x|^2 &\stackrel{\mathcal{L}}{=}&  \alpha^2((T+(T\xi QZ_1/2)^{1/2}S)^2+Z_2Q/2)
\end{array}
\end{equation}
where $\stackrel{\mathcal{L}}{=}$ represents the equality in distribution.
When $T$ goes to infinity, the law of large numbers states that $\frac{1}{T}Q\xrightarrow{a.s.}2$ and $T^{1+\epsilon} Z_1\xrightarrow{a.s.}0$ for any $\epsilon>0$, then
\begin{equation}
\begin{array}{lll}
\frac{1}{T}\|\hat\z\|^2 &\xrightarrow{a.s.}& 1,\\
\frac{1}{T}|\hat\z^H\x| &\xrightarrow{a.s.}&  \alpha .
\end{array}
\end{equation}
In that case, the information density satisfies
\begin{eqnarray}
&&(1-\frac{1}{T})\log_2(2\beta T) + 2\beta(\alpha-1)\log_2(e)-\frac{\log_2(\Gamma(T+1))}{T}\nonumber\\
&&-\frac{1}{T}i(\x;\hat\z)\xrightarrow{T\rightarrow\infty} 0
\end{eqnarray}
where we used $\frac{1}{T}\log_2 I_0(T) \xrightarrow{T\rightarrow\infty} \log_2( e)$.

\end{document}

%% file: mse_T3200_hist1.tikz
% Minimal pgfplots version: 1.3
%  http://www.mathworks.com/matlabcentral/fileexchange/22022-matlab2tikz
\definecolor{mycolor1}{rgb}{1.00000,0.00000,1.00000}%
\usetikzlibrary{calc}
\begin{tikzpicture}

\begin{axis}[%
width=2.5142969in,
height=1.561771in,
at={(1.224167in,0.803229in)},
scale only axis,
xmin=-40,
xmax=7,
xlabel={SNR (dB)},
xmajorgrids,
ymin=-40,
ymax=3,
ylabel={MSE (dB)},
ymajorgrids,
axis x line*=bottom,
axis y line*=left,
legend style={at={(0.505468,0.59591)},anchor=south west,legend cell align=left,align=left,draw=white!15!black}
]

%\addplot [color=blue,solid,mark=x,mark options={solid}]
%  table[row sep=crcr]{%
%-40	      0.4507\\
%-35.05	-1.7014\\
%-33	  -2.9063\\
%-30	  -5.0546\\
%-27	  -7.5355\\
%-25.05	-9.3164\\
%-23 	-11.1879\\
%-15.05	-18.9750\\
%-5.05 	-28.9268\\
%4.94	-38.8613\\
%};
%\addlegendentry{Lower bound \eqref{eq:cr_variance}  for TBM (64 50)};

\addplot [color=blue,solid,mark=x,mark options={solid}]
  table[row sep=crcr]{%
  -40	0.413649207164781\\
-35.05	-1.74254788644248\\
-33	-2.98169942465671\\
-30	-5.14412288038778\\
-27	-7.6388330974147\\
-25.05	-9.38611617960026\\
-23	-11.2949133269198\\
-15.05	-19.0385860241126\\
-5.05	-29.0018946312522\\
4.94	-38.9882055205809\\
};
\addlegendentry{Approximate lower bound from \eqref{eq:mse_approx} \& \eqref{eq:mse_lb}};

\addplot [color=blue,solid,mark=triangle,mark repeat=10,mark options={solid}]
  table[row sep=crcr]{%
-40	0\\
-39.7738693467337	0\\
-39.5477386934673	0\\
-39.321608040201	0\\
-39.0954773869347	0\\
-38.8693467336683	0\\
-38.643216080402	0\\
-38.4170854271357	0\\
-38.1909547738693	0\\
-37.964824120603	0\\
-37.7386934673367	0\\
-37.5125628140704	0\\
-37.286432160804	0\\
-37.0603015075377	0\\
-36.8341708542714	0\\
-36.608040201005	0\\
-36.3819095477387	0\\
-36.1557788944724	0\\
-35.929648241206	0\\
-35.7035175879397	0\\
-35.4773869346734	0\\
-35.251256281407	0\\
-35.0251256281407	0\\
-34.7989949748744	0\\
-34.572864321608	0\\
-34.3467336683417	0\\
-34.1206030150754	0\\
-33.894472361809	0\\
-33.6683417085427	0\\
-33.4422110552764	0\\
-33.2160804020101	0\\
-32.9899497487437	0\\
-32.7638190954774	0\\
-32.5376884422111	0\\
-32.3115577889447	0\\
-32.0854271356784	0\\
-31.8592964824121	0\\
-31.6331658291457	0\\
-31.4070351758794	0\\
-31.1809045226131	0\\
-30.9547738693467	0\\
-30.7286432160804	0\\
-30.5025125628141	0\\
-30.2763819095477	0\\
-30.0502512562814	0\\
-29.8241206030151	0\\
-29.5979899497487	0\\
-29.3718592964824	0\\
-29.1457286432161	0\\
-28.9195979899498	0\\
-28.6934673366834	0\\
-28.4673366834171	-4.36216422960172\\
-28.2412060301508	-4.9477011591224\\
-28.0150753768844	-5.42131361619016\\
-27.7889447236181	-5.83981251642047\\
-27.5628140703518	-6.22401042700926\\
-27.3366834170854	-6.58430246199089\\
-27.1105527638191	-6.92676252128115\\
-26.8844221105528	-7.2553020521695\\
-26.6582914572864	-7.57261390851903\\
-26.4321608040201	-7.88064493247407\\
-26.2060301507538	-8.18085633996844\\
-25.9798994974874	-8.47437784554489\\
-25.7537688442211	-8.76210412922927\\
-25.5276381909548	-9.04475798411431\\
-25.3015075376884	-9.32293320474002\\
-25.0753768844221	-9.59712462874336\\
-24.8492462311558	-9.86774973930994\\
-24.6231155778894	-10.135164554215\\
-24.3969849246231	-10.3996755448485\\
-24.1708542713568	-10.6615487332627\\
-23.9447236180905	-10.921016742758\\
-23.7185929648241	-11.1782843378097\\
-23.4924623115578	-11.4335328310062\\
-23.2663316582915	-11.6869236280232\\
-23.0402010050251	-11.9386011082955\\
-22.8140703517588	-12.1886949876702\\
-22.5879396984925	-12.4373222727467\\
-22.3618090452261	-12.6845888901915\\
-22.1356783919598	-12.9305910549657\\
-21.9095477386935	-13.1754164270568\\
-21.6834170854271	-13.4191450955474\\
-21.4572864321608	-13.6618504206988\\
-21.2311557788945	-13.9035997584821\\
-21.0050251256281	-14.1444550871674\\
-20.7788944723618	-14.3844735518233\\
-20.5527638190955	-14.623707939628\\
-20.3266331658291	-14.862207096555\\
-20.1005025125628	-15.1000162941392\\
-19.8743718592965	-15.3371775535311\\
-19.6482412060302	-15.5737299328469\\
-19.4221105527638	-15.809709782838\\
-19.1959798994975	-16.0451509751097\\
-18.9698492462312	-16.2800851064589\\
-18.7437185929648	-16.5145416823618\\
-18.5175879396985	-16.7485482821933\\
-18.2914572864322	-16.9821307083861\\
-18.0653266331658	-17.2153131214266\\
-17.8391959798995	-17.4481181623191\\
-17.6130653266332	-17.680567063934\\
-17.3869346733668	-17.9126797524622\\
-17.1608040201005	-18.1444749400452\\
-16.9346733668342	-18.3759702095089\\
-16.7085427135678	-18.6071820920193\\
-16.4824120603015	-18.8381261383731\\
-16.2562814070352	-19.0688169845558\\
-16.0301507537688	-19.2992684121215\\
-15.8040201005025	-19.5294934038881\\
-15.5778894472362	-19.7595041953848\\
-15.3517587939699	-19.989312322439\\
-15.1256281407035	-20.2189286652523\\
-14.8994974874372	-20.4483634892694\\
-14.6733668341709	-20.6776264831239\\
-14.4472361809045	-20.9067267939048\\
-14.2211055276382	-21.1356730599683\\
-13.9949748743719	-21.364473441498\\
-13.7688442211055	-21.5931356489929\\
-13.5427135678392	-21.8216669698498\\
-13.3165829145729	-22.0500742931878\\
-13.0904522613065	-22.27836413305\\
-12.8643216080402	-22.5065426501075\\
-12.6381909547739	-22.734615671973\\
-12.4120603015075	-22.9625887122317\\
-12.1859296482412	-23.1904669882788\\
-11.9597989949749	-23.41825543805\\
-11.7336683417085	-23.645958735724\\
-11.5075376884422	-23.8735813064674\\
-11.2814070351759	-24.1011273402881\\
-11.0552763819095	-24.3286008050602\\
-10.8291457286432	-24.5560054587712\\
-10.6030150753769	-24.7833448610477\\
-10.3768844221106	-25.0106223840066\\
-10.1507537688442	-25.2378412224704\\
-9.92462311557789	-25.4650044035957\\
-9.69849246231156	-25.6921147959468\\
-9.47236180904522	-25.919175118048\\
-9.24623115577889	-26.1461879464552\\
-9.02010050251256	-26.3731557233644\\
-8.79396984924623	-26.6000807637988\\
-8.5678391959799	-26.8269652623911\\
-8.34170854271357	-27.0538112997845\\
-8.11557788944724	-27.2806208486872\\
-7.88944723618091	-27.5073957795831\\
-7.66331658291458	-27.7341378661355\\
-7.43718592964824	-27.9608487902919\\
-7.21105527638191	-28.1875301471107\\
-6.98492462311557	-28.4141834493254\\
-6.75879396984924	-28.6408101316623\\
-6.53266331658291	-28.8674115549222\\
-6.30653266331658	-29.0939890098478\\
-6.08040201005025	-29.3205437207766\\
-5.85427135678392	-29.5470768491045\\
-5.62814070351759	-29.7735894965601\\
-5.40201005025126	-30.0000827083067\\
-5.17587939698493	-30.2265574758763\\
-4.94974874371859	-30.453014739958\\
-4.72361809045226	-30.679455393022\\
-4.49748743718593	-30.9058802818217\\
-4.2713567839196	-31.1322902097539\\
-4.04522613065327	-31.3586859390977\\
-3.81909547738694	-31.5850681931389\\
-3.5929648241206	-31.8114376581762\\
-3.36683417085427	-32.0377949854286\\
-3.14070351758794	-32.264140792843\\
-2.91457286432161	-32.4904756668025\\
-2.68844221105527	-32.7168001637471\\
-2.46231155778894	-32.9431148117187\\
-2.23618090452261	-33.1694201118138\\
-2.01005025125628	-33.3957165395604\\
-1.78391959798995	-33.6220045462477\\
-1.55778894472362	-33.8482845601432\\
-1.33165829145729	-34.0745569876905\\
-1.10552763819096	-34.3008222146229\\
-0.879396984924625	-34.5270806070149\\
-0.653266331658294	-34.7533325122989\\
-0.427135678391963	-34.9795782602108\\
-0.201005025125632	-35.2058181636972\\
0.0251256281407066	-35.4320525197751\\
0.251256281407038	-35.658281610343\\
0.477386934673369	-35.8845057029465\\
0.7035175879397	-36.1107250515223\\
0.929648241206031	-36.3369398970942\\
1.15577889447236	-36.5631504684154\\
1.38190954773869	-36.7893569826097\\
1.60804020100503	-37.0155596457558\\
1.83417085427136	-37.2417586534495\\
2.06030150753769	-37.4679541913561\\
2.28643216080402	-37.6941464356848\\
2.51256281407035	-37.9203355537013\\
2.73869346733668	-38.1465217041554\\
2.96482412060301	-38.3727050377466\\
3.19095477386934	-38.598885697497\\
3.41708542713568	-38.8250638191795\\
3.64321608040201	-39.0512395316455\\
3.86934673366834	-39.2774129572249\\
4.09547738693468	-39.5035842120138\\
4.32160804020101	-39.7297534062193\\
4.54773869346734	-39.9559206444471\\
4.77386934673367	-40.1820860260032\\
5	-40.408249645114\\
};
\addlegendentry{AMP theoretical performance \cite{kadmon18}};

\addplot [color=blue,dashed,mark=o,mark options={solid}]
  table[row sep=crcr]{%
  -40.0000    2.4904\\
  -35.0500    2.4751\\
  -33.0000    2.4551\\
  -30.0000    2.3147\\
  -29.0000    2.0516\\
  -28.0000    1.2825\\
  -27.0000   -0.4084\\
  -26.0000   -3.2325\\
  -25.0500   -6.6443\\
  -15.0500  -18.8700\\
   -5.0500  -28.9107\\
    4.9400  -38.9043\\
};
\addlegendentry{GN algorithm \cite{sorber13} with random initialization};

\addplot [color=blue,dashed,mark=square,mark options={solid}]
  table[row sep=crcr]{%
 -40.0000    2.0810\\
  -35.0500    1.8341\\
  -33.0000    1.5406\\
  -30.0000   -0.2086\\
  -29.0000   -1.8537\\
  -28.0000   -3.9261\\
  -27.0000   -5.7376\\
  -26.0000   -7.1501\\
  -25.0500   -8.3203\\
  -15.0500  -18.8701\\
   -5.0500  -28.9111\\
    4.9400  -38.9051\\
	};
\addlegendentry{GN algorithm \cite{sorber13} with genie initialization};

\coordinate (P30) at (axis cs:-30, 2.3147);
\coordinate (P29) at (axis cs:-29, 2.0516);
\coordinate (P28) at (axis cs:-28, 1.2825);
\coordinate (P27) at (axis cs:-27, -0.4084);
\coordinate (P26) at (axis cs:-26, -3.2325);
\coordinate (P25) at (axis cs:-25, -6.6443);
\end{axis}

\def\x{24} % reference of histograms
\begin{axis}[%
width=2.5142969in,
height=1.561771in,
at={(1.224167in,2.603229in)},
scale only axis,
xmin=-6,
xmax=0,
xticklabels={-31,-30,...,-\x},
xlabel={SNR (dB)},
xmajorgrids,
ymin=-15,
ymax=5,
ylabel={MSE (dB)},
ymajorgrids,
axis x line*=bottom,
axis y line*=left,
legend style={at={(0.525468,0.59591)},anchor=south west,legend cell align=left,align=left,draw=white!15!black}
]

\coordinate (Q30) at (axis cs:-30+\x, 2.3147);
\coordinate (Q29) at (axis cs:-29+\x,  2.0516);
\coordinate (Q28) at (axis cs:-28+\x, 1.2825);
\coordinate (Q27) at (axis cs:-27+\x, -0.4084);
\coordinate (Q26) at (axis cs:-26+\x, -3.2325);
\coordinate (Q25) at (axis cs:-25+\x, -6.6443);

\addplot[xbar,bar width=0.001cm,draw=blue,axis y line = none,shift={(axis direction cs:\x-40,0)}] coordinates {
(0.002000,0.750000)
(0.000000,0.800000)
(0.000000,0.850000)
(0.000000,0.900000)
(0.002000,0.950000)
(0.002000,1.000000)
(0.000000,1.050000)
(0.004000,1.100000)
(0.004000,1.150000)
(0.006000,1.200000)
(0.006000,1.250000)
(0.004000,1.300000)
(0.008000,1.350000)
(0.018000,1.400000)
(0.026000,1.450000)
(0.024000,1.500000)
(0.032000,1.550000)
(0.070000,1.600000)
(0.072000,1.650000)
(0.074000,1.700000)
(0.116000,1.750000)
(0.138000,1.800000)
(0.196000,1.850000)
(0.258000,1.900000)
(0.274000,1.950000)
(0.380000,2.000000)
(0.468000,2.050000)
(0.504000,2.100000)
(0.538000,2.150000)
(0.754000,2.200000)
(0.884000,2.250000)
(0.986000,2.300000)
(1.150000,2.350000)
(1.180000,2.400000)
(1.252000,2.450000)
(1.400000,2.500000)
(1.416000,2.550000)
(1.428000,2.600000)
(1.340000,2.650000)
(1.364000,2.700000)
(1.068000,2.750000)
(1.008000,2.800000)
(0.752000,2.850000)
(0.512000,2.900000)
(0.270000,2.950000)
(0.010000,3.000000)
 };

\addplot[xbar,bar width=0.001cm,draw=blue,fill=blue,axis y line = none,shift={(axis direction cs:\x-35,0)}] coordinates {
(0.002000,0.750000)
(0.000000,0.800000)
(0.000000,0.850000)
(0.000000,0.900000)
(0.002000,0.950000)
(0.002000,1.000000)
(0.000000,1.050000)
(0.004000,1.100000)
(0.004000,1.150000)
(0.006000,1.200000)
(0.006000,1.250000)
(0.004000,1.300000)
(0.008000,1.350000)
(0.018000,1.400000)
(0.026000,1.450000)
(0.024000,1.500000)
(0.032000,1.550000)
(0.070000,1.600000)
(0.072000,1.650000)
(0.074000,1.700000)
(0.116000,1.750000)
(0.138000,1.800000)
(0.196000,1.850000)
(0.258000,1.900000)
(0.274000,1.950000)
(0.380000,2.000000)
(0.468000,2.050000)
(0.504000,2.100000)
(0.538000,2.150000)
(0.754000,2.200000)
(0.884000,2.250000)
(0.986000,2.300000)
(1.150000,2.350000)
(1.180000,2.400000)
(1.252000,2.450000)
(1.400000,2.500000)
(1.416000,2.550000)
(1.428000,2.600000)
(1.340000,2.650000)
(1.364000,2.700000)
(1.068000,2.750000)
(1.008000,2.800000)
(0.752000,2.850000)
(0.512000,2.900000)
(0.270000,2.950000)
(0.010000,3.000000)
 };

\addplot[xbar,bar width=0.001cm,draw=blue,fill=blue,axis y line = none,shift={(axis direction cs:\x-33,0)}] coordinates {
(0.002212,0.250000)
(0.002212,0.300000)
(0.000000,0.350000)
(0.002212,0.400000)
(0.000000,0.450000)
(0.000000,0.500000)
(0.002212,0.550000)
(0.002212,0.600000)
(0.000000,0.650000)
(0.002212,0.700000)
(0.000000,0.750000)
(0.002212,0.800000)
(0.002212,0.850000)
(0.004425,0.900000)
(0.011062,0.950000)
(0.002212,1.000000)
(0.015487,1.050000)
(0.004425,1.100000)
(0.013274,1.150000)
(0.008850,1.200000)
(0.011062,1.250000)
(0.026549,1.300000)
(0.024336,1.350000)
(0.044248,1.400000)
(0.057522,1.450000)
(0.059735,1.500000)
(0.064159,1.550000)
(0.088496,1.600000)
(0.128319,1.650000)
(0.139381,1.700000)
(0.196903,1.750000)
(0.194690,1.800000)
(0.294248,1.850000)
(0.283186,1.900000)
(0.396018,1.950000)
(0.491150,2.000000)
(0.564159,2.050000)
(0.674779,2.100000)
(0.727876,2.150000)
(0.845133,2.200000)
(0.993363,2.250000)
(1.137168,2.300000)
(1.163717,2.350000)
(1.265487,2.400000)
(1.305310,2.450000)
(1.431416,2.500000)
(1.444690,2.550000)
(1.391593,2.600000)
(1.500000,2.650000)
(1.336283,2.700000)
(1.227876,2.750000)
(1.015487,2.800000)
(0.745575,2.850000)
(0.484513,2.900000)
(0.283186,2.950000)
(0.008850,3.000000)
};

\addplot[xbar,bar width=0.02cm,draw=blue,fill=blue,axis y line = none,shift={(axis direction cs:\x-30,0)}] coordinates {
(0.000796,-6.200000)
(0.000000,-6.100000)
(0.000000,-6.000000)
(0.000796,-5.900000)
(0.000000,-5.800000)
(0.000000,-5.700000)
(0.000000,-5.600000)
(0.000000,-5.500000)
(0.000000,-5.400000)
(0.000000,-5.300000)
(0.000796,-5.200000)
(0.002387,-5.100000)
(0.000796,-5.000000)
(0.000796,-4.900000)
(0.000796,-4.800000)
(0.002387,-4.700000)
(0.000000,-4.600000)
(0.000796,-4.500000)
(0.002387,-4.400000)
(0.002387,-4.300000)
(0.001592,-4.200000)
(0.000796,-4.100000)
(0.003979,-4.000000)
(0.000796,-3.900000)
(0.004775,-3.800000)
(0.001592,-3.700000)
(0.002387,-3.600000)
(0.001592,-3.500000)
(0.001592,-3.400000)
(0.003979,-3.300000)
(0.000796,-3.200000)
(0.001592,-3.100000)
(0.000796,-3.000000)
(0.003979,-2.900000)
(0.000796,-2.800000)
(0.003979,-2.700000)
(0.001592,-2.600000)
(0.002387,-2.500000)
(0.006366,-2.400000)
(0.000796,-2.300000)
(0.003183,-2.200000)
(0.002387,-2.100000)
(0.001592,-2.000000)
(0.004775,-1.900000)
(0.003979,-1.800000)
(0.003183,-1.700000)
(0.003183,-1.600000)
(0.002387,-1.500000)
(0.003183,-1.400000)
(0.003979,-1.300000)
(0.002387,-1.200000)
(0.003183,-1.100000)
(0.001592,-1.000000)
(0.003979,-0.900000)
(0.003979,-0.800000)
(0.003183,-0.700000)
(0.006366,-0.600000)
(0.005570,-0.500000)
(0.005570,-0.400000)
(0.000796,-0.300000)
(0.004775,-0.200000)
(0.004775,-0.100000)
(0.004775,0.000000)
(0.011141,0.100000)
(0.010345,0.200000)
(0.005570,0.300000)
(0.005570,0.400000)
(0.007958,0.500000)
(0.014324,0.600000)
(0.020690,0.700000)
(0.025464,0.800000)
(0.023873,0.900000)
(0.023873,1.000000)
(0.034218,1.100000)
(0.059682,1.200000)
(0.054907,1.300000)
(0.071618,1.400000)
(0.089920,1.500000)
(0.144032,1.600000)
(0.173475,1.700000)
(0.229973,1.800000)
(0.302387,1.900000)
(0.426525,2.000000)
(0.520424,2.100000)
(0.674801,2.200000)
(0.739257,2.300000)
(0.849867,2.400000)
(0.900000,2.500000)
(0.884085,2.600000)
(0.773475,2.700000)
(0.521220,2.800000)
(0.212467,2.900000)
(0.002387,3.000000)
};

\addplot[xbar,bar width=0.06cm,draw=blue,fill=blue,axis y line = none,shift={(axis direction cs:\x-29,0)}] coordinates {
(0.000306,-7.500000)
(0.000612,-7.200000)
(0.000000,-6.900000)
(0.001529,-6.600000)
(0.003058,-6.300000)
(0.008869,-6.000000)
(0.011315,-5.700000)
(0.011009,-5.400000)
(0.014679,-5.100000)
(0.015291,-4.800000)
(0.022936,-4.500000)
(0.020489,-4.200000)
(0.016208,-3.900000)
(0.014373,-3.600000)
(0.014067,-3.300000)
(0.011927,-3.000000)
(0.011009,-2.700000)
(0.011927,-2.400000)
(0.009786,-2.100000)
(0.008257,-1.800000)
(0.008563,-1.500000)
(0.005199,-1.200000)
(0.009786,-0.900000)
(0.007339,-0.600000)
(0.008563,-0.300000)
(0.015902,0.000000)
(0.022018,0.300000)
(0.032722,0.600000)
(0.048012,0.900000)
(0.081346,1.200000)
(0.181957,1.500000)
(0.370031,1.800000)
(0.671254,2.100000)
(0.900000,2.400000)
(0.487156,2.700000)
(0.000612,3.000000)
 };
 
\addplot[xbar,bar width=0.11cm,draw=blue,fill=blue,axis y line = none,shift={(axis direction cs:\x-28,0)}] coordinates {
(0.000914,-8.500000)
(0.006091,-8.000000)
(0.021015,-7.500000)
(0.050558,-7.000000)
(0.079492,-6.500000)
(0.126701,-6.000000)
(0.139492,-5.500000)
(0.114518,-5.000000)
(0.080711,-4.500000)
(0.051472,-4.000000)
(0.035025,-3.500000)
(0.018579,-3.000000)
(0.011878,-2.500000)
(0.010660,-2.000000)
(0.012487,-1.500000)
(0.012183,-1.000000)
(0.018579,-0.500000)
(0.027107,0.000000)
(0.043249,0.500000)
(0.114213,1.000000)
(0.347208,1.500000)
(0.900000,2.000000)
(0.823249,2.500000)
(0.000305,3.000000)
};

\addplot[xbar,bar width=0.11cm,draw=blue,fill=blue,axis y line = none,shift={(axis direction cs:\x-27,0)}] coordinates {
(0.005011,-9.500000)
(0.026058,-9.000000)
(0.068151,-8.500000)
(0.185412,-8.000000)
(0.364811,-7.500000)
(0.509131,-7.000000)
(0.525167,-6.500000)
(0.429955,-6.000000)
(0.273608,-5.500000)
(0.148831,-5.000000)
(0.067149,-4.500000)
(0.030568,-4.000000)
(0.019042,-3.500000)
(0.007016,-3.000000)
(0.009521,-2.500000)
(0.006013,-2.000000)
(0.007016,-1.500000)
(0.013029,-1.000000)
(0.017038,-0.500000)
(0.029566,0.000000)
(0.052116,0.500000)
(0.116759,1.000000)
(0.367817,1.500000)
(0.900000,2.000000)
(0.831849,2.500000)
(0.000501,3.000000)
};

\addplot[xbar,bar width=0.11cm,draw=blue,fill=blue,axis y line = none,shift={(axis direction cs:\x-26,0)}] coordinates {
(0.000536,-11.000000)
(0.007504,-10.500000)
(0.050387,-10.000000)
(0.125968,-9.500000)
(0.355926,-9.000000)
(0.618582,-8.500000)
(0.889815,-8.000000)
(0.900000,-7.500000)
(0.681298,-7.000000)
(0.396665,-6.500000)
(0.188148,-6.000000)
(0.060036,-5.500000)
(0.024122,-5.000000)
(0.011793,-4.500000)
(0.003216,-4.000000)
(0.002680,-3.500000)
(0.000536,-3.000000)
(0.002144,-2.500000)
(0.002680,-2.000000)
(0.003216,-1.500000)
(0.004824,-1.000000)
(0.008577,-0.500000)
(0.007504,0.000000)
(0.020905,0.500000)
(0.062716,1.000000)
(0.161346,1.500000)
(0.383264,2.000000)
(0.385944,2.500000)
};

\addplot[xbar,bar width=0.06cm,draw=blue,fill=blue,axis y line = none,shift={(axis direction cs:\x-25,0)}] coordinates {
(0.002819,-11.700000)
(0.004229,-11.400000)
(0.022553,-11.100000)
(0.053563,-10.800000)
(0.088802,-10.500000)
(0.200861,-10.200000)
(0.348160,-9.900000)
(0.491229,-9.600000)
(0.680110,-9.300000)
(0.807674,-9.000000)
(0.900000,-8.700000)
(0.829522,-8.400000)
(0.766797,-8.100000)
(0.594127,-7.800000)
(0.386922,-7.500000)
(0.242443,-7.200000)
(0.143070,-6.900000)
(0.075411,-6.600000)
(0.045810,-6.300000)
(0.023962,-6.000000)
(0.009162,-5.700000)
(0.002114,-5.400000)
(0.001410,-5.100000)
(0.000000,-4.800000)
(0.000000,-4.500000)
(0.000000,-4.200000)
(0.000000,-3.900000)
(0.000000,-3.600000)
(0.000000,-3.300000)
(0.000000,-3.000000)
(0.000705,-2.700000)
(0.000705,-2.400000)
(0.000705,-2.100000)
(0.000705,-1.800000)
(0.000705,-1.500000)
(0.002819,-1.200000)
(0.000705,-0.900000)
(0.000705,-0.600000)
(0.000000,-0.300000)
(0.000000,0.000000)
(0.001410,0.300000)
(0.002819,0.600000)
(0.007048,0.900000)
(0.014800,1.200000)
(0.020439,1.500000)
(0.033829,1.800000)
(0.082459,2.100000)
(0.109945,2.400000)
(0.046515,2.700000)
};

\addplot[xbar,bar width=0.005cm,draw=blue,fill=blue,axis y line = none,shift={(axis direction cs:\x-15,0)}] coordinates {
(0.002988,-22.100000)
(0.005976,-22.000000)
(0.000000,-21.900000)
(0.005976,-21.800000)
(0.011952,-21.700000)
(0.005976,-21.600000)
(0.017928,-21.500000)
(0.029880,-21.400000)
(0.032869,-21.300000)
(0.047809,-21.200000)
(0.059761,-21.100000)
(0.077689,-21.000000)
(0.104582,-20.900000)
(0.101594,-20.800000)
(0.197211,-20.700000)
(0.197211,-20.600000)
(0.289841,-20.500000)
(0.358566,-20.400000)
(0.421315,-20.300000)
(0.466135,-20.200000)
(0.519920,-20.100000)
(0.684263,-20.000000)
(0.845618,-19.900000)
(0.872510,-19.800000)
(1.039841,-19.700000)
(1.201195,-19.600000)
(1.117530,-19.500000)
(1.287849,-19.400000)
(1.443227,-19.300000)
(1.416335,-19.200000)
(1.380478,-19.100000)
(1.368526,-19.000000)
(1.413347,-18.900000)
(1.500000,-18.800000)
(1.341633,-18.700000)
(1.323705,-18.600000)
(1.129482,-18.500000)
(1.105578,-18.400000)
(1.075697,-18.300000)
(0.905378,-18.200000)
(0.732072,-18.100000)
(0.663347,-18.000000)
(0.582669,-17.900000)
(0.519920,-17.800000)
(0.445219,-17.700000)
(0.352590,-17.600000)
(0.274900,-17.500000)
(0.206175,-17.400000)
(0.164343,-17.300000)
(0.101594,-17.200000)
(0.122510,-17.100000)
(0.062749,-17.000000)
(0.077689,-16.900000)
(0.047809,-16.800000)
(0.029880,-16.700000)
(0.014940,-16.600000)
(0.035857,-16.500000)
(0.002988,-16.400000)
(0.014940,-16.300000)
(0.005976,-16.200000)
(0.005976,-16.100000)
(0.000000,-16.000000)
(0.005976,-15.900000)
(0.000000,-15.800000)
(0.000000,-15.700000)
(0.002988,-15.600000)
};

\addplot[xbar,bar width=0.03cm,draw=blue,fill=blue,axis y line = none,shift={(axis direction cs:\x-5,0)}] coordinates {
(0.005837,-32.100000)
(0.002918,-32.000000)
(0.002918,-31.900000)
(0.011673,-31.800000)
(0.011673,-31.700000)
(0.014591,-31.600000)
(0.008755,-31.500000)
(0.029183,-31.400000)
(0.046693,-31.300000)
(0.035019,-31.200000)
(0.070039,-31.100000)
(0.078794,-31.000000)
(0.107977,-30.900000)
(0.125486,-30.800000)
(0.224708,-30.700000)
(0.215953,-30.600000)
(0.285992,-30.500000)
(0.353113,-30.400000)
(0.382296,-30.300000)
(0.501946,-30.200000)
(0.609922,-30.100000)
(0.732490,-30.000000)
(0.820039,-29.900000)
(0.890078,-29.800000)
(1.091440,-29.700000)
(1.152724,-29.600000)
(1.199416,-29.500000)
(1.249027,-29.400000)
(1.438716,-29.300000)
(1.500000,-29.200000)
(1.257782,-29.100000)
(1.333658,-29.000000)
(1.456226,-28.900000)
(1.348249,-28.800000)
(1.362840,-28.700000)
(1.243191,-28.600000)
(1.068093,-28.500000)
(1.015564,-28.400000)
(0.986381,-28.300000)
(0.817121,-28.200000)
(0.732490,-28.100000)
(0.633268,-28.000000)
(0.528210,-27.900000)
(0.475681,-27.800000)
(0.420233,-27.700000)
(0.309339,-27.600000)
(0.245136,-27.500000)
(0.160506,-27.400000)
(0.137160,-27.300000)
(0.113813,-27.200000)
(0.081712,-27.100000)
(0.064202,-27.000000)
(0.072957,-26.900000)
(0.026265,-26.800000)
(0.020428,-26.700000)
(0.023346,-26.600000)
(0.014591,-26.500000)
(0.008755,-26.400000)
(0.008755,-26.300000)
(0.005837,-26.200000)
(0.005837,-26.100000)
(0.000000,-26.000000)
(0.002918,-25.900000)
(0.000000,-25.800000)
(0.002918,-25.700000)
};

\addplot[xbar,bar width=0.03cm,draw=blue,fill=blue,axis y line = none,shift={(axis direction cs:\x+5,0)}] coordinates {
(0.005917,-42.100000)
(0.002959,-42.000000)
(0.002959,-41.900000)
(0.014793,-41.800000)
(0.008876,-41.700000)
(0.017751,-41.600000)
(0.005917,-41.500000)
(0.032544,-41.400000)
(0.032544,-41.300000)
(0.050296,-41.200000)
(0.071006,-41.100000)
(0.062130,-41.000000)
(0.130178,-40.900000)
(0.115385,-40.800000)
(0.198225,-40.700000)
(0.236686,-40.600000)
(0.275148,-40.500000)
(0.331361,-40.400000)
(0.402367,-40.300000)
(0.532544,-40.200000)
(0.618343,-40.100000)
(0.718935,-40.000000)
(0.834320,-39.900000)
(0.917160,-39.800000)
(1.053254,-39.700000)
(1.168639,-39.600000)
(1.281065,-39.500000)
(1.245562,-39.400000)
(1.346154,-39.300000)
(1.500000,-39.200000)
(1.366864,-39.100000)
(1.366864,-39.000000)
(1.491124,-38.900000)
(1.334320,-38.800000)
(1.363905,-38.700000)
(1.284024,-38.600000)
(1.097633,-38.500000)
(1.044379,-38.400000)
(0.970414,-38.300000)
(0.801775,-38.200000)
(0.801775,-38.100000)
(0.621302,-38.000000)
(0.544379,-37.900000)
(0.497041,-37.800000)
(0.437870,-37.700000)
(0.337278,-37.600000)
(0.224852,-37.500000)
(0.201183,-37.400000)
(0.121302,-37.300000)
(0.124260,-37.200000)
(0.082840,-37.100000)
(0.062130,-37.000000)
(0.062130,-36.900000)
(0.044379,-36.800000)
(0.017751,-36.700000)
(0.023669,-36.600000)
(0.008876,-36.500000)
(0.011834,-36.400000)
(0.008876,-36.300000)
(0.005917,-36.200000)
(0.002959,-36.100000)
(0.002959,-36.000000)
(0.002959,-35.900000)
(0.002959,-35.800000)
};
\end{axis}
\draw [dashed] (P30) -- (Q30);
\draw [dashed] (P29) -- (Q29);
\draw [dashed] (P28) -- (Q28);
\draw [dashed] (P27) -- (Q27);
\draw [dashed] (P26) -- (Q26);
\draw [dashed] (P25) -- (Q25);
\end{tikzpicture}%

%% file: mse_T3200_hist100.tikz
% Minimal pgfplots version: 1.3
%  http://www.mathworks.com/matlabcentral/fileexchange/22022-matlab2tikz
\definecolor{mycolor1}{rgb}{1.00000,0.00000,1.00000}%
\usetikzlibrary{calc}
\begin{tikzpicture}

\begin{axis}[%
width=2.5142969in,
height=1.961771in,
at={(1.224167in,0.803229in)},
scale only axis,
xmin=-40,
xmax=7,
xlabel={SNR (dB)},
xmajorgrids,
ymin=-50,
ymax=10,
ylabel={MSE (dB)},
ymajorgrids,
axis x line*=bottom,
axis y line*=left,
legend style={at={(0.525468,0.59591)},anchor=south west,legend cell align=left,align=left,draw=white!15!black}
]

%\addplot [color=blue,solid,mark=x,mark options={solid}]
%  table[row sep=crcr]{%
%-40	      0.5044\\
%-35.05	-1.5878\\
%-33	  -2.7985\\
%-30	  -4.9228\\
%-27 	-7.3946\\
%-25.05	-9.1757\\
%-23  	-11.0208\\
%-15.05	-18.7905\\
%-5.05	-28.7691\\
%4.94	-38.7053\\
%};
%\addlegendentry{Lower bound \eqref{eq:cr_variance}  for TBM (64 50)};

\addplot [color=blue,solid,mark=x,mark options={solid}]
  table[row sep=crcr]{%
-40	0.469978690071002\\
-35.05	-1.64666946528513\\
-33	-2.86829139276279\\
-30	-5.00821597001948\\
-27	-7.4866376674201\\
-25.05	-9.22670620490231\\
-23	-11.1302250294069\\
-15.05	-18.8657869377504\\
-5.05	-28.8276014457121\\
4.94	-38.8137612221825\\
};
\addlegendentry{Approximate lower bound from \eqref{eq:mse_approx} \& \eqref{eq:mse_lb}};

\addplot [color=blue,solid,mark=triangle,mark repeat=10,mark options={solid}]
  table[row sep=crcr]{%
-40	0\\
-39.7738693467337	0\\
-39.5477386934673	0\\
-39.321608040201	0\\
-39.0954773869347	0\\
-38.8693467336683	0\\
-38.643216080402	0\\
-38.4170854271357	0\\
-38.1909547738693	0\\
-37.964824120603	0\\
-37.7386934673367	0\\
-37.5125628140704	0\\
-37.286432160804	0\\
-37.0603015075377	0\\
-36.8341708542714	0\\
-36.608040201005	0\\
-36.3819095477387	0\\
-36.1557788944724	0\\
-35.929648241206	0\\
-35.7035175879397	0\\
-35.4773869346734	0\\
-35.251256281407	0\\
-35.0251256281407	0\\
-34.7989949748744	0\\
-34.572864321608	0\\
-34.3467336683417	0\\
-34.1206030150754	0\\
-33.894472361809	0\\
-33.6683417085427	0\\
-33.4422110552764	0\\
-33.2160804020101	0\\
-32.9899497487437	0\\
-32.7638190954774	0\\
-32.5376884422111	0\\
-32.3115577889447	0\\
-32.0854271356784	0\\
-31.8592964824121	0\\
-31.6331658291457	0\\
-31.4070351758794	0\\
-31.1809045226131	0\\
-30.9547738693467	0\\
-30.7286432160804	0\\
-30.5025125628141	0\\
-30.2763819095477	0\\
-30.0502512562814	0\\
-29.8241206030151	0\\
-29.5979899497487	0\\
-29.3718592964824	0\\
-29.1457286432161	0\\
-28.9195979899498	0\\
-28.6934673366834	0\\
-28.4673366834171	-4.36216422960172\\
-28.2412060301508	-4.9477011591224\\
-28.0150753768844	-5.42131361619016\\
-27.7889447236181	-5.83981251642047\\
-27.5628140703518	-6.22401042700926\\
-27.3366834170854	-6.58430246199089\\
-27.1105527638191	-6.92676252128115\\
-26.8844221105528	-7.2553020521695\\
-26.6582914572864	-7.57261390851903\\
-26.4321608040201	-7.88064493247407\\
-26.2060301507538	-8.18085633996844\\
-25.9798994974874	-8.47437784554489\\
-25.7537688442211	-8.76210412922927\\
-25.5276381909548	-9.04475798411431\\
-25.3015075376884	-9.32293320474002\\
-25.0753768844221	-9.59712462874336\\
-24.8492462311558	-9.86774973930994\\
-24.6231155778894	-10.135164554215\\
-24.3969849246231	-10.3996755448485\\
-24.1708542713568	-10.6615487332627\\
-23.9447236180905	-10.921016742758\\
-23.7185929648241	-11.1782843378097\\
-23.4924623115578	-11.4335328310062\\
-23.2663316582915	-11.6869236280232\\
-23.0402010050251	-11.9386011082955\\
-22.8140703517588	-12.1886949876702\\
-22.5879396984925	-12.4373222727467\\
-22.3618090452261	-12.6845888901915\\
-22.1356783919598	-12.9305910549657\\
-21.9095477386935	-13.1754164270568\\
-21.6834170854271	-13.4191450955474\\
-21.4572864321608	-13.6618504206988\\
-21.2311557788945	-13.9035997584821\\
-21.0050251256281	-14.1444550871674\\
-20.7788944723618	-14.3844735518233\\
-20.5527638190955	-14.623707939628\\
-20.3266331658291	-14.862207096555\\
-20.1005025125628	-15.1000162941392\\
-19.8743718592965	-15.3371775535311\\
-19.6482412060302	-15.5737299328469\\
-19.4221105527638	-15.809709782838\\
-19.1959798994975	-16.0451509751097\\
-18.9698492462312	-16.2800851064589\\
-18.7437185929648	-16.5145416823618\\
-18.5175879396985	-16.7485482821933\\
-18.2914572864322	-16.9821307083861\\
-18.0653266331658	-17.2153131214266\\
-17.8391959798995	-17.4481181623191\\
-17.6130653266332	-17.680567063934\\
-17.3869346733668	-17.9126797524622\\
-17.1608040201005	-18.1444749400452\\
-16.9346733668342	-18.3759702095089\\
-16.7085427135678	-18.6071820920193\\
-16.4824120603015	-18.8381261383731\\
-16.2562814070352	-19.0688169845558\\
-16.0301507537688	-19.2992684121215\\
-15.8040201005025	-19.5294934038881\\
-15.5778894472362	-19.7595041953848\\
-15.3517587939699	-19.989312322439\\
-15.1256281407035	-20.2189286652523\\
-14.8994974874372	-20.4483634892694\\
-14.6733668341709	-20.6776264831239\\
-14.4472361809045	-20.9067267939048\\
-14.2211055276382	-21.1356730599683\\
-13.9949748743719	-21.364473441498\\
-13.7688442211055	-21.5931356489929\\
-13.5427135678392	-21.8216669698498\\
-13.3165829145729	-22.0500742931878\\
-13.0904522613065	-22.27836413305\\
-12.8643216080402	-22.5065426501075\\
-12.6381909547739	-22.734615671973\\
-12.4120603015075	-22.9625887122317\\
-12.1859296482412	-23.1904669882788\\
-11.9597989949749	-23.41825543805\\
-11.7336683417085	-23.645958735724\\
-11.5075376884422	-23.8735813064674\\
-11.2814070351759	-24.1011273402881\\
-11.0552763819095	-24.3286008050602\\
-10.8291457286432	-24.5560054587712\\
-10.6030150753769	-24.7833448610477\\
-10.3768844221106	-25.0106223840066\\
-10.1507537688442	-25.2378412224704\\
-9.92462311557789	-25.4650044035957\\
-9.69849246231156	-25.6921147959468\\
-9.47236180904522	-25.919175118048\\
-9.24623115577889	-26.1461879464552\\
-9.02010050251256	-26.3731557233644\\
-8.79396984924623	-26.6000807637988\\
-8.5678391959799	-26.8269652623911\\
-8.34170854271357	-27.0538112997845\\
-8.11557788944724	-27.2806208486872\\
-7.88944723618091	-27.5073957795831\\
-7.66331658291458	-27.7341378661355\\
-7.43718592964824	-27.9608487902919\\
-7.21105527638191	-28.1875301471107\\
-6.98492462311557	-28.4141834493254\\
-6.75879396984924	-28.6408101316623\\
-6.53266331658291	-28.8674115549222\\
-6.30653266331658	-29.0939890098478\\
-6.08040201005025	-29.3205437207766\\
-5.85427135678392	-29.5470768491045\\
-5.62814070351759	-29.7735894965601\\
-5.40201005025126	-30.0000827083067\\
-5.17587939698493	-30.2265574758763\\
-4.94974874371859	-30.453014739958\\
-4.72361809045226	-30.679455393022\\
-4.49748743718593	-30.9058802818217\\
-4.2713567839196	-31.1322902097539\\
-4.04522613065327	-31.3586859390977\\
-3.81909547738694	-31.5850681931389\\
-3.5929648241206	-31.8114376581762\\
-3.36683417085427	-32.0377949854286\\
-3.14070351758794	-32.264140792843\\
-2.91457286432161	-32.4904756668025\\
-2.68844221105527	-32.7168001637471\\
-2.46231155778894	-32.9431148117187\\
-2.23618090452261	-33.1694201118138\\
-2.01005025125628	-33.3957165395604\\
-1.78391959798995	-33.6220045462477\\
-1.55778894472362	-33.8482845601432\\
-1.33165829145729	-34.0745569876905\\
-1.10552763819096	-34.3008222146229\\
-0.879396984924625	-34.5270806070149\\
-0.653266331658294	-34.7533325122989\\
-0.427135678391963	-34.9795782602108\\
-0.201005025125632	-35.2058181636972\\
0.0251256281407066	-35.4320525197751\\
0.251256281407038	-35.658281610343\\
0.477386934673369	-35.8845057029465\\
0.7035175879397	-36.1107250515223\\
0.929648241206031	-36.3369398970942\\
1.15577889447236	-36.5631504684154\\
1.38190954773869	-36.7893569826097\\
1.60804020100503	-37.0155596457558\\
1.83417085427136	-37.2417586534495\\
2.06030150753769	-37.4679541913561\\
2.28643216080402	-37.6941464356848\\
2.51256281407035	-37.9203355537013\\
2.73869346733668	-38.1465217041554\\
2.96482412060301	-38.3727050377466\\
3.19095477386934	-38.598885697497\\
3.41708542713568	-38.8250638191795\\
3.64321608040201	-39.0512395316455\\
3.86934673366834	-39.2774129572249\\
4.09547738693468	-39.5035842120138\\
4.32160804020101	-39.7297534062193\\
4.54773869346734	-39.9559206444471\\
4.77386934673367	-40.1820860260032\\
5	-40.408249645114\\
};
\addlegendentry{AMP theoretical performance \cite{kadmon18}};

\addplot [color=blue,dashed,mark=o,mark options={solid}]
  table[row sep=crcr]{%
  -40.0000    2.0444\\
  -35.0500    1.9460\\
  -33.0000    1.8564\\
  -30.0000    1.5958\\
  -29.0000    1.4403\\
  -28.0000    1.2024\\
  -27.0000    0.7892\\
  -26.0000   -0.1168\\
  -25.0500   -2.6100\\
  -15.0500  -18.2651\\
   -5.0500  -27.4415\\
    4.9400  -33.2409\\
};
\addlegendentry{GN algorithm \cite{sorber13} with random initialization};

\addplot [color=blue,dashed,mark=square,mark options={solid}]
  table[row sep=crcr]{%
  -40.0000    1.6137\\
  -35.0500    1.2092\\
  -33.0000    0.8161\\
  -30.0000   -0.4633\\
  -29.0000   -1.2580\\
  -28.0000   -2.3736\\
  -27.0000   -3.9070\\
  -26.0000   -5.7264\\
  -25.0500   -7.2469\\
  -15.0500  -18.3046\\
   -5.0500  -28.3741\\
    4.9400  -38.3705\\
};
\addlegendentry{GN algorithm \cite{sorber13} with genie initialization};

\coordinate (P30) at (axis cs:-30, 1.5958);
\coordinate (P29) at (axis cs:-29, 1.4403);
\coordinate (P28) at (axis cs:-28, 1.2024);
\coordinate (P27) at (axis cs:-27, 0.7892);
\coordinate (P26) at (axis cs:-26, -0.1168);
\coordinate (P25) at (axis cs:-25, -2.6100);
\end{axis}

\def\x{24} % reference of histograms

\begin{axis}[%
width=2.5142969in,
height=1.961771in,
at={(1.224167in,3.203229in)},
scale only axis,
xmin=-6,
xmax=0,
xticklabels={-31,-30,...,-\x},
xlabel={SNR (dB)},
xmajorgrids,
ymin=-15,
ymax=5,
ylabel={MSE (dB)},
ymajorgrids,
axis x line*=bottom,
axis y line*=left,
legend style={at={(0.525468,0.59591)},anchor=south west,legend cell align=left,align=left,draw=white!15!black}
]

\coordinate (Q30) at (axis cs:-30+\x, 1.5958);
\coordinate (Q29) at (axis cs:-29+\x,  1.4403);
\coordinate (Q28) at (axis cs:-28+\x, 1.2024);
\coordinate (Q27) at (axis cs:-27+\x, 0.7892);
\coordinate (Q26) at (axis cs:-26+\x, -0.1168);
\coordinate (Q25) at (axis cs:-25+\x, -2.6100);

\addplot[xbar,bar width=0.001cm,draw=blue,axis y line = none,shift={(axis direction cs:\x-40,0)}] coordinates {
(0.000194,0.050000)
(0.000000,0.110000)
(0.000000,0.170000)
(0.000000,0.230000)
(0.000000,0.290000)
(0.000194,0.350000)
(0.000194,0.410000)
(0.000583,0.470000)
(0.000194,0.530000)
(0.002526,0.590000)
(0.001748,0.650000)
(0.006605,0.710000)
(0.005828,0.770000)
(0.009519,0.830000)
(0.013599,0.890000)
(0.021176,0.950000)
(0.028558,1.010000)
(0.041963,1.070000)
(0.058865,1.130000)
(0.086647,1.190000)
(0.123559,1.250000)
(0.160666,1.310000)
(0.229633,1.370000)
(0.281894,1.430000)
(0.380197,1.490000)
(0.461793,1.550000)
(0.614104,1.610000)
(0.723676,1.670000)
(0.885119,1.730000)
(1.051807,1.790000)
(1.193628,1.850000)
(1.305142,1.910000)
(1.420347,1.970000)
(1.476493,2.030000)
(1.500000,2.090000)
(1.452014,2.150000)
(1.343803,2.210000)
(1.161961,2.270000)
(0.980508,2.330000)
(0.765251,2.390000)
(0.576415,2.450000)
(0.399819,2.510000)
(0.280339,2.570000)
(0.169214,2.630000)
(0.100635,2.690000)
(0.059254,2.750000)
(0.029141,2.810000)
(0.015153,2.870000)
(0.007188,2.930000)
(0.000389,2.990000)
 };

\addplot[xbar,bar width=0.001cm,draw=blue,fill=blue,axis y line = none,shift={(axis direction cs:\x-35,0)}] coordinates {
(0.000206,-0.120000)
(0.000000,-0.057000)
(0.000000,0.006000)
(0.000823,0.069000)
(0.000206,0.132000)
(0.000411,0.195000)
(0.001029,0.258000)
(0.002057,0.321000)
(0.002057,0.384000)
(0.002263,0.447000)
(0.004115,0.510000)
(0.005555,0.573000)
(0.011727,0.636000)
(0.014196,0.699000)
(0.019956,0.762000)
(0.029008,0.825000)
(0.038678,0.888000)
(0.055342,0.951000)
(0.087437,1.014000)
(0.117679,1.077000)
(0.155534,1.140000)
(0.213551,1.203000)
(0.286586,1.266000)
(0.368468,1.329000)
(0.483267,1.392000)
(0.586751,1.455000)
(0.727472,1.518000)
(0.878892,1.581000)
(1.004595,1.644000)
(1.175559,1.707000)
(1.283569,1.770000)
(1.406186,1.833000)
(1.472226,1.896000)
(1.499177,1.959000)
(1.500000,2.022000)
(1.453710,2.085000)
(1.283569,2.148000)
(1.134001,2.211000)
(0.920450,2.274000)
(0.729735,2.337000)
(0.523591,2.400000)
(0.384515,2.463000)
(0.266630,2.526000)
(0.171787,2.589000)
(0.112125,2.652000)
(0.072418,2.715000)
(0.046084,2.778000)
(0.026128,2.841000)
(0.011933,2.904000)
(0.002057,2.967000)
 };

\addplot[xbar,bar width=0.001cm,draw=blue,fill=blue,axis y line = none,shift={(axis direction cs:\x-33,0)}] coordinates {
(0.000208,-0.360000)
(0.000208,-0.292000)
(0.000208,-0.224000)
(0.000208,-0.156000)
(0.001455,-0.088000)
(0.001039,-0.020000)
(0.001039,0.048000)
(0.002078,0.116000)
(0.002494,0.184000)
(0.004780,0.252000)
(0.006234,0.320000)
(0.008520,0.388000)
(0.012884,0.456000)
(0.015378,0.524000)
(0.028055,0.592000)
(0.039277,0.660000)
(0.056110,0.728000)
(0.070865,0.796000)
(0.108687,0.864000)
(0.137573,0.932000)
(0.172485,1.000000)
(0.227764,1.068000)
(0.311305,1.136000)
(0.404821,1.204000)
(0.496052,1.272000)
(0.615337,1.340000)
(0.767041,1.408000)
(0.932461,1.476000)
(1.045511,1.544000)
(1.182045,1.612000)
(1.319618,1.680000)
(1.414796,1.748000)
(1.500000,1.816000)
(1.486284,1.884000)
(1.487531,1.952000)
(1.380507,2.020000)
(1.252909,2.088000)
(1.076475,2.156000)
(0.867207,2.224000)
(0.700748,2.292000)
(0.547382,2.360000)
(0.384871,2.428000)
(0.270366,2.496000)
(0.174979,2.564000)
(0.110349,2.632000)
(0.071488,2.700000)
(0.048836,2.768000)
(0.022236,2.836000)
(0.011222,2.904000)
(0.001455,2.972000)
};

\addplot[xbar,bar width=0.03cm,draw=blue,fill=blue,axis y line = none,shift={(axis direction cs:\x-30,0)}] coordinates {
(0.000136,-3.500000)
(0.000000,-3.337000)
(0.000000,-3.174000)
(0.000000,-3.011000)
(0.000068,-2.848000)
(0.000068,-2.685000)
(0.000136,-2.522000)
(0.000339,-2.359000)
(0.000407,-2.196000)
(0.000136,-2.033000)
(0.000339,-1.870000)
(0.000679,-1.707000)
(0.001018,-1.544000)
(0.001358,-1.381000)
(0.001086,-1.218000)
(0.002444,-1.055000)
(0.003802,-0.892000)
(0.005160,-0.729000)
(0.007196,-0.566000)
(0.008961,-0.403000)
(0.016361,-0.240000)
(0.025458,-0.077000)
(0.041208,0.086000)
(0.061303,0.249000)
(0.099796,0.412000)
(0.153971,0.575000)
(0.245214,0.738000)
(0.369722,0.901000)
(0.521724,1.064000)
(0.685743,1.227000)
(0.840869,1.390000)
(0.900000,1.553000)
(0.877529,1.716000)
(0.744671,1.879000)
(0.545282,2.042000)
(0.334080,2.205000)
(0.169993,2.368000)
(0.082349,2.531000)
(0.032179,2.694000)
(0.008079,2.857000)
};

\addplot[xbar,bar width=0.04cm,draw=blue,fill=blue,axis y line = none,shift={(axis direction cs:\x-29,0)}] coordinates {
(0.000065,-4.600000)
(0.000131,-4.410000)
(0.000131,-4.220000)
(0.000261,-4.030000)
(0.000131,-3.840000)
(0.000392,-3.650000)
(0.000784,-3.460000)
(0.000588,-3.270000)
(0.001242,-3.080000)
(0.001438,-2.890000)
(0.001830,-2.700000)
(0.001700,-2.510000)
(0.002223,-2.320000)
(0.003661,-2.130000)
(0.003530,-1.940000)
(0.004511,-1.750000)
(0.006080,-1.560000)
(0.006864,-1.370000)
(0.011048,-1.180000)
(0.013925,-0.990000)
(0.019547,-0.800000)
(0.023273,-0.610000)
(0.037917,-0.420000)
(0.055633,-0.230000)
(0.081194,-0.040000)
(0.116757,0.150000)
(0.180431,0.340000)
(0.267575,0.530000)
(0.388451,0.720000)
(0.530508,0.910000)
(0.708782,1.100000)
(0.839725,1.290000)
(0.900000,1.480000)
(0.821094,1.670000)
(0.648834,1.860000)
(0.429309,2.050000)
(0.248355,2.240000)
(0.119438,2.430000)
(0.046611,2.620000)
(0.013402,2.810000)
 };
 
\addplot[xbar,bar width=0.05cm,draw=blue,fill=blue,axis y line = none,shift={(axis direction cs:\x-28,0)}] coordinates {
(0.000063,-6.600000)
(0.000000,-6.359000)
(0.000126,-6.118000)
(0.000251,-5.877000)
(0.000377,-5.636000)
(0.000690,-5.395000)
(0.000628,-5.154000)
(0.001318,-4.913000)
(0.002636,-4.672000)
(0.003138,-4.431000)
(0.004456,-4.190000)
(0.004895,-3.949000)
(0.005334,-3.708000)
(0.007405,-3.467000)
(0.008472,-3.226000)
(0.012426,-2.985000)
(0.012238,-2.744000)
(0.014748,-2.503000)
(0.015501,-2.262000)
(0.021777,-2.021000)
(0.024726,-1.780000)
(0.032132,-1.539000)
(0.038533,-1.298000)
(0.049955,-1.057000)
(0.064891,-0.816000)
(0.086730,-0.575000)
(0.124071,-0.334000)
(0.179046,-0.093000)
(0.256865,0.148000)
(0.376857,0.389000)
(0.521386,0.630000)
(0.693090,0.871000)
(0.843142,1.112000)
(0.900000,1.353000)
(0.799024,1.594000)
(0.578181,1.835000)
(0.349745,2.076000)
(0.164486,2.317000)
(0.063071,2.558000)
(0.013305,2.799000)
};

\addplot[xbar,bar width=0.05cm,draw=blue,fill=blue,axis y line = none,shift={(axis direction cs:\x-27,0)}] coordinates {
(0.000075,-7.600000)
(0.000151,-7.334000)
(0.001132,-7.068000)
(0.001509,-6.802000)
(0.002037,-6.536000)
(0.005131,-6.270000)
(0.008149,-6.004000)
(0.011393,-5.738000)
(0.017279,-5.472000)
(0.022485,-5.206000)
(0.028446,-4.940000)
(0.031539,-4.674000)
(0.035161,-4.408000)
(0.043536,-4.142000)
(0.048139,-3.876000)
(0.045423,-3.610000)
(0.050931,-3.344000)
(0.055684,-3.078000)
(0.058853,-2.812000)
(0.066474,-2.546000)
(0.065342,-2.280000)
(0.072359,-2.014000)
(0.085488,-1.748000)
(0.103672,-1.482000)
(0.120800,-1.216000)
(0.152037,-0.950000)
(0.190594,-0.684000)
(0.242203,-0.418000)
(0.324975,-0.152000)
(0.418461,0.114000)
(0.565443,0.380000)
(0.698541,0.646000)
(0.827716,0.912000)
(0.900000,1.178000)
(0.848466,1.444000)
(0.651232,1.710000)
(0.423592,1.976000)
(0.217757,2.242000)
(0.083526,2.508000)
(0.019542,2.774000)
};

\addplot[xbar,bar width=0.07cm,draw=blue,fill=blue,axis y line = none,shift={(axis direction cs:\x-26,0)}] coordinates {
(0.000210,-9.300000)
(0.000420,-8.992000)
(0.001996,-8.684000)
(0.003782,-8.376000)
(0.006723,-8.068000)
(0.016809,-7.760000)
(0.032882,-7.452000)
(0.056204,-7.144000)
(0.091607,-6.836000)
(0.125225,-6.528000)
(0.155795,-6.220000)
(0.190674,-5.912000)
(0.208007,-5.604000)
(0.214311,-5.296000)
(0.215887,-4.988000)
(0.211369,-4.680000)
(0.193930,-4.372000)
(0.198342,-4.064000)
(0.179118,-3.756000)
(0.175336,-3.448000)
(0.171344,-3.140000)
(0.163464,-2.832000)
(0.175441,-2.524000)
(0.185736,-2.216000)
(0.197292,-1.908000)
(0.227968,-1.600000)
(0.256963,-1.292000)
(0.300560,-0.984000)
(0.364433,-0.676000)
(0.451103,-0.368000)
(0.547648,-0.060000)
(0.681277,0.248000)
(0.792740,0.556000)
(0.900000,0.864000)
(0.882561,1.172000)
(0.789903,1.480000)
(0.587989,1.788000)
(0.358550,2.096000)
(0.157371,2.404000)
(0.034458,2.712000)
};

\addplot[xbar,bar width=0.1cm,draw=blue,fill=blue,axis y line = none,shift={(axis direction cs:\x-25,0)}] coordinates {
(0.000120,-11.100000)
(0.000600,-10.747000)
(0.002519,-10.394000)
(0.008995,-10.041000)
(0.024107,-9.688000)
(0.068843,-9.335000)
(0.156397,-8.982000)
(0.294803,-8.629000)
(0.471828,-8.276000)
(0.656410,-7.923000)
(0.825520,-7.570000)
(0.900000,-7.217000)
(0.863539,-6.864000)
(0.781503,-6.511000)
(0.636980,-6.158000)
(0.492457,-5.805000)
(0.372641,-5.452000)
(0.297681,-5.099000)
(0.245389,-4.746000)
(0.200653,-4.393000)
(0.175346,-4.040000)
(0.159395,-3.687000)
(0.148961,-3.334000)
(0.149440,-2.981000)
(0.154478,-2.628000)
(0.163713,-2.275000)
(0.169590,-1.922000)
(0.176546,-1.569000)
(0.206170,-1.216000)
(0.216844,-0.863000)
(0.252945,-0.510000)
(0.294803,-0.157000)
(0.333782,0.196000)
(0.370482,0.549000)
(0.399147,0.902000)
(0.422295,1.255000)
(0.381397,1.608000)
(0.290005,1.961000)
(0.179184,2.314000)
(0.048094,2.667000)
};

\addplot[xbar,bar width=0.02cm,draw=blue,fill=blue,axis y line = none,shift={(axis direction cs:\x-15,0)}] coordinates {
(0.000325,-22.200000)
(0.005477,-21.570000)
(0.052545,-20.940000)
(0.265653,-20.310000)
(0.832098,-19.680000)
(1.476249,-19.050000)
(1.500000,-18.420000)
(0.900043,-17.790000)
(0.311149,-17.160000)
(0.068433,-16.530000)
(0.008513,-15.900000)
(0.000868,-15.270000)
(0.000000,-14.640000)
(0.000000,-14.010000)
(0.000000,-13.380000)
(0.000000,-12.750000)
(0.000000,-12.120000)
(0.000000,-11.490000)
(0.000000,-10.860000)
(0.000000,-10.230000)
(0.000000,-9.600000)
(0.000000,-8.970000)
(0.000000,-8.340000)
(0.000000,-7.710000)
(0.000000,-7.080000)
(0.000000,-6.450000)
(0.000108,-5.820000)
(0.000217,-5.190000)
(0.000000,-4.560000)
(0.000271,-3.930000)
(0.000271,-3.300000)
(0.000217,-2.670000)
(0.000000,-2.040000)
(0.000000,-1.410000)
(0.000000,-0.780000)
(0.000000,-0.150000)
(0.000000,0.480000)
(0.000054,1.110000)
(0.000054,1.740000)
(0.000054,2.370000)
};

\addplot[xbar,bar width=0.02cm,draw=blue,fill=blue,axis y line = none,shift={(axis direction cs:\x-5,0)}] coordinates {
(0.000038,-32.800000)
(0.003885,-31.900000)
(0.083620,-31.000000)
(0.608377,-30.100000)
(1.500000,-29.200000)
(1.241141,-28.300000)
(0.359711,-27.400000)
(0.041618,-26.500000)
(0.003692,-25.600000)
(0.001154,-24.700000)
(0.000538,-23.800000)
(0.000308,-22.900000)
(0.000154,-22.000000)
(0.000192,-21.100000)
(0.000038,-20.200000)
(0.000000,-19.300000)
(0.000000,-18.400000)
(0.000000,-17.500000)
(0.000000,-16.600000)
(0.000000,-15.700000)
(0.000000,-14.800000)
(0.000038,-13.900000)
(0.000038,-13.000000)
(0.000038,-12.100000)
(0.000000,-11.200000)
(0.000038,-10.300000)
(0.000115,-9.400000)
(0.000077,-8.500000)
(0.000077,-7.600000)
(0.000077,-6.700000)
(0.000115,-5.800000)
(0.000115,-4.900000)
(0.000192,-4.000000)
(0.000231,-3.100000)
(0.000077,-2.200000)
(0.000231,-1.300000)
(0.000000,-0.400000)
(0.000038,0.500000)
(0.000115,1.400000)
(0.000269,2.300000)
};

\addplot[xbar,bar width=0.02cm,draw=blue,fill=blue,axis y line = none,shift={(axis direction cs:\x+5,0)}] coordinates {
(0.005444,-42.000000)
(0.190782,-40.870000)
(1.204576,-39.740000)
(1.500000,-38.610000)
(0.382331,-37.480000)
(0.028056,-36.350000)
(0.005812,-35.220000)
(0.005210,-34.090000)
(0.003874,-32.960000)
(0.003006,-31.830000)
(0.002538,-30.700000)
(0.002204,-29.570000)
(0.001870,-28.440000)
(0.000935,-27.310000)
(0.000601,-26.180000)
(0.000701,-25.050000)
(0.000234,-23.920000)
(0.000134,-22.790000)
(0.000100,-21.660000)
(0.000067,-20.530000)
(0.000033,-19.400000)
(0.000100,-18.270000)
(0.000000,-17.140000)
(0.000000,-16.010000)
(0.000000,-14.880000)
(0.000033,-13.750000)
(0.000000,-12.620000)
(0.000100,-11.490000)
(0.000033,-10.360000)
(0.000067,-9.230000)
(0.000134,-8.100000)
(0.000067,-6.970000)
(0.000067,-5.840000)
(0.000200,-4.710000)
(0.000134,-3.580000)
(0.000100,-2.450000)
(0.000033,-1.320000)
(0.000033,-0.190000)
(0.000033,0.940000)
(0.000367,2.070000)
};
\end{axis}
\draw [dashed] (P30) -- (Q30);
\draw [dashed] (P29) -- (Q29);
\draw [dashed] (P28) -- (Q28);
\draw [dashed] (P27) -- (Q27);
\draw [dashed] (P26) -- (Q26);
\draw [dashed] (P25) -- (Q25);

\end{tikzpicture}%

%% file: blervssnr300_full.tikz
% Minimal pgfplots version: 1.3
%  http://www.mathworks.com/matlabcentral/fileexchange/22022-matlab2tikz
\definecolor{mycolor1}{rgb}{0.00000,0.44700,0.74100}%
\definecolor{mycolor2}{rgb}{0.85000,0.32500,0.09800}%
\definecolor{mycolor3}{rgb}{0.92900,0.69400,0.12500}%
\definecolor{mycolor4}{rgb}{0.49400,0.18400,0.55600}%
\definecolor{mycolor5}{rgb}{0.46600,0.67400,0.18800}%
\begin{tikzpicture}

\begin{axis}[
  hide axis,
  legend columns=-1,
  legend style={at={(0.35,1.15)},anchor=south west,legend cell align=left,align=left,draw=white!15!black,nodes={scale=0.8, transform shape},/tikz/every even column/.append style={column sep=0.5cm}}
  ]
% dummy plot to make legend show up
\addplot[forget plot] coordinates {(0,0)};

\addlegendimage{only marks,mark=o,mark options={solid}}
\addlegendentry{$\Ka=1$}

\addlegendimage{only marks,mark=x,mark options={solid}}
\addlegendentry{$\Ka=100$}
\addlegendimage{only marks,mark=triangle,mark options={solid}}
\addlegendentry{$\Ka=200$}

\end{axis}

\begin{axis}[%
width=2.820833in,
height=2.065625in,
at={(0.758333in,0.48125in)},
scale only axis,
xmin=-27,
xmax=-16,
xlabel={SNR},
xmajorgrids,
ymode=log,
ymin=1e-06,
ymax=1,
ylabel={Average PER},
ymajorgrids,
yminorgrids,
title style={font=\bfseries},
axis x line*=bottom,
axis y line*=left,
legend style={at={(0.04,1.11)},anchor=south west,legend cell align=left,align=left,draw=white!15!black,nodes={scale=0.8, transform shape}}
]
% SNR = -20dB with 200 users

\addlegendimage{color=mycolor4,line legend,no markers}
\addlegendentry{DT bound \eqref{eq:dt} for channel \eqref{eq:equivalent_channel0} and equivalent SNR \eqref{eq:snr_eq}};

%\addlegendimage{color=mycolor3,line legend,no markers}
%\addlegendentry{DT bound \eqref{eq:dt} for channel \eqref{eq:equivalent_channel0} and MSE \eqref{eq:amp_mse}};

\addlegendimage{color=mycolor1,line legend,no markers}
\addlegendentry{TBM, GN CPD and soft demapping, polar code (sub-constellation (i))};

\addlegendimage{color=mycolor2,line legend}
%\addlegendentry{TBM, GN CPD and soft demapping, polar code CS no face bits};
\addlegendentry{TBM, GN CPD and soft demapping, polar code (sub-constellation (ii))};

\addlegendimage{color=mycolor5,line legend}
\addlegendentry{TBM, GN CPD and hard demapping, BCH code};

\addlegendimage{color=mycolor1,line legend,dashed}
\addlegendentry{Equivalent channel \eqref{eq:equivalent_channel0} and soft demapping, polar code  (sub-constellation (i))};

\addlegendimage{color=mycolor2,line legend,dashed}
%\addlegendentry{Equivalent channel \eqref{eq:equivalent_channel0} and soft demapping, polar code CS no face bits};
\addlegendentry{Equivalent channel \eqref{eq:equivalent_channel0} and soft demapping, polar code  (sub-constellation (ii))};

%\addlegendimage{color=mycolor2,line legend,dash dot}
%\addlegendentry{Equivalent channel \eqref{eq:equivalent_channel0} and soft demapping, polar code (rate 0.67) CS no face bits PA};

%\addlegendimage{color=mycolor2,line legend,dash dot}
%\addlegendentry{Equivalent channel \eqref{eq:equivalent_channel0} and soft demapping, polar code (rate 0.67) QAM+pilot};

\addlegendimage{color=mycolor5,line legend,dashed}
\addlegendentry{Equivalent channel \eqref{eq:equivalent_channel0} and hard demapping, BCH};

\addplot [color=mycolor4,mark=o,mark options={solid}]
  table[row sep=crcr]{%
-28	  1\\
-27	  0.8389\\
-26	  0.1275\\
-25     0.0017\\
-24    2.1221e-13\\ 
-23    0\\% < 1e-16  
};
%\addlegendentry{DT bound \eqref{eq:dt} for channel \eqref{eq:equivalent_channel0} and equivalent SNR \eqref{eq:snr_eq} -- $\Ka=1$};

\addplot [color=mycolor4,mark=x,mark options={solid}]
  table[row sep=crcr]{%
-28	  1\\
-27	  0.9578\\
-26	  0.2112\\
-25     0.0029\\
-24    9.5645e-12\\ 
-23    0\\% < 1e-16  
};
%\addlegendentry{DT bound \eqref{eq:dt} for channel \eqref{eq:equivalent_channel0} and equivalent SNR \eqref{eq:snr_eq} -- $\Ka=100$};

\addplot [color=mycolor4,mark=triangle,mark options={solid}]
  table[row sep=crcr]{%
-28	  1\\
-27	  0.9081\\
-26	 0.3179\\
-25     0.0067\\
-24    8.5102e-07\\ 
-23    0\\% < 1e-16  
};

\addplot [color=mycolor1,mark=o,mark options={solid}]
  table[row sep=crcr]{%
-28	  1\\
-27	  1\\
-25	  0.96\\
-23   0.28\\
-21   0.0820\\
-19   0.0180\\
-17   0.0015\\
-16   2.9e-4\\
};
%\addlegendentry{TBM, GN CPD and soft demapping, polar code -- $\Ka=1$};

\addplot [color=mycolor1,mark=x,mark options={solid}]
  table[row sep=crcr]{%
-28	  1\\
-27	  1\\
-25	  0.9990\\
-23   0.5626\\
-21   0.0748\\
-19   0.0110\\
-17   0.00084\\
-16   1.25e-4\\
};
%\addlegendentry{TBM, GN CPD and soft demapping, polar code -- $\Ka=100$};

\addplot [color=mycolor1,mark=triangle,mark options={solid}]
  table[row sep=crcr]{%
-28	  1\\
-27	  1\\
-25	  1\\
-23   1\\
-21   0.1057\\
-19   0.0134\\
-17   0.0013\\
-16   2e-4\\
};
%\addlegendentry{TBM, GN CPD and soft demapping, polar code -- $\Ka=200$};

%\addplot [color=mycolor2,dashed,mark=o,mark options={solid}]
%  table[row sep=crcr]{%
%-28	  1\\
%-27	  1\\
%-25	  0.81\\
%-24	  0.23\\
%-23   6e-2\\
%-22   1e-3\\
%-21   1.1e-5\\ % < 0.5e-5
%};
%\addlegendentry{Equivalent channel \eqref{eq:equivalent_channel0} and soft demapping, polar code  CS no face bits -- $\Ka=1$};

%\addplot [color=mycolor2,dash dot,mark=o,mark options={solid}]
%  table[row sep=crcr]{%
%-27	  1\\%?
%-26	  0.81\\%?
%-25	  0.27\\ %?
%-24	 3.2e-2\\
%-23   1.1e-3\\
%-22   2e-5\\
%};
%\addlegendentry{Equivalent channel \eqref{eq:equivalent_channel0} and soft demapping, polar code (rate 0.67) CS no face bits PA};

\addplot [color=mycolor2,dashed,mark=o,mark options={solid}]
  table[row sep=crcr]{%
-27	  1\\
-26	  1\\
-25	  0.89\\ 
-24	 0.53\\
-23   0.13\\
-22   0.017\\
-21   6.4e-4\\
-20   1e-5\\
};
%\addlegendentry{Equivalent channel \eqref{eq:equivalent_channel0} and soft demapping, polar code - QAM and random gain K_a=1};

\addplot [color=mycolor2,dashed,mark=x,mark options={solid}]
  table[row sep=crcr]{%
-27	  1\\
-26	  1\\
-25	  0.90\\ 
-24	 0.58\\
-23   0.207\\
-22   0.026\\
-21   9.0e-4\\
-20   1.9e-5\\
};
%\addlegendentry{Equivalent channel \eqref{eq:equivalent_channel0} and soft demapping, polar code - QAM and random gain K_a=100};

\addplot [color=mycolor2,dashed,mark=triangle,mark options={solid}]
  table[row sep=crcr]{%
-27	  1\\
-26	  1\\
-25   0.95\\
-24   0.6797\\
-23   2.455e-1\\
-22   3.75e-2\\
-21   1.92e-3\\
-20   4.3e-5\\
};
%\addlegendentry{Equivalent channel \eqref{eq:equivalent_channel0} and soft demapping, polar code - QAM and random gain K_a=200};

%\addplot [color=mycolor2,dash dot,mark=*,mark options={solid}]
%  table[row sep=crcr]{%
%-27	  0.995\\
%-26	  0.93\\
%-25	  0.445\\ 
%-24	  5e-2\\
%-23   8.90e-4\\ 
%-22   5e-6\\
%};
%\addlegendentry{Equivalent channel \eqref{eq:equivalent_channel0} and soft demapping, polar code (rate 0.67) QAM+pilot};

%\addplot [color=mycolor2,mark=o,mark options={solid}]
%  table[row sep=crcr]{%
%-28	  1\\
%-27	  1\\
%-25	  0.97\\
%-24	  0.78\\
%-23   0.35\\
%-22   5e-2\\
%-21   3.4e-3\\ % < 0.5e-5
%};
%\addlegendentry{TBM, GN CPD and soft demapping, polar code CS no face bits -- $\Ka=1$};

\addplot [color=mycolor2,mark=o,mark options={solid}]
  table[row sep=crcr]{%
-28	  1\\
-27	  1\\
-25	  0.98\\
-24	  0.84\\
-23   0.378\\
-22   7.8e-2\\
-21   5.9e-3\\
-20   1.4e-4\\ 
};
%\addlegendentry{TBM, GN CPD and soft demapping, polar code QAM -- $\Ka=1$};

\addplot [color=mycolor2,mark=x,mark options={solid}]
  table[row sep=crcr]{%
-28	  1\\
-25	  1\\
-24	  0.99\\
-23	  0.80\\
-22	  0.21\\
-21	  0.018\\
-20	  7.4e-4\\
-19	  1.33e-5\\
};
%\addlegendentry{TBM, GN CPD and soft demapping, polar code QAM -- $\Ka=100$};

\addplot [color=mycolor2,mark=triangle,mark options={solid}]
  table[row sep=crcr]{%
-28	  1\\
-25	  1\\
-24	  1\\
-23	  1\\
-22	  1\\
-21	0.092\\
-20	0.0036\\
-19	1.8e-4\\
};
%\addlegendentry{TBM, GN CPD and soft demapping, polar code QAM -- $\Ka=200$};

\addplot [color=mycolor1,dashed,mark=o,mark options={solid}]
  table[row sep=crcr]{%
-28	  1\\
-27	  1\\
-25	  0.8120\\
-23   0.1940\\
-21   0.05160\\
-19   0.00580\\
-17   0.000500\\
-16   0.8e-4\\
};
%\addlegendentry{Equivalent channel \eqref{eq:equivalent_channel0} and soft demapping, polar code -- $\Ka=1$};

\addplot [color=mycolor1,dashed,mark=x,mark options={solid}]
  table[row sep=crcr]{%
-28	  1\\
-27	  1\\
-25	  0.8940\\
-23   0.2340\\
-21   0.0520\\
-19   0.0070\\
-17   0.000550\\ 
-16   0.000110\\ 
};
%\addlegendentry{Equivalent channel \eqref{eq:equivalent_channel0} and soft demapping, polar code -- $\Ka=100$};

\addplot [color=mycolor1,dashed,mark=triangle,mark options={solid}]
  table[row sep=crcr]{%
-28	  1\\
-27	  1\\
-25	  0.9000\\
-23   0.2340\\
-21   0.0670\\
-19   0.006250\\ 
-18   0.002750\\
-17   0.000650\\ 
-16   1.65e-4\\
};
%\addlegendentry{Equivalent channel \eqref{eq:equivalent_channel0} and soft demapping, polar code -- $\Ka=200$};

\addplot [color=mycolor5,dashed,mark=o,mark options={solid,rotate=180}]
  table[row sep=crcr]{%
-28	  1\\
-25	  1\\
-23   0.9292\\ 
-21   0.3965\\ 
-19   0.034\\ 
-17   3e-4\\ 
-16   1e-5\\ 
};
%\addlegendentry{Equivalent channel \eqref{eq:equivalent_channel0} and hard demapping, BCH -- $\Ka=1$};

\addplot [color=mycolor5,dashed,mark=x,mark options={solid,rotate=180}]
  table[row sep=crcr]{%
-28	  1\\
-25	  1\\
-23   0.9480\\ 
-21   0.4560\\ 
-19   0.049\\ 
-17   6e-4\\ 
-16   1.8e-5\\ 
};
%\addlegendentry{Equivalent channel \eqref{eq:equivalent_channel0} and hard demapping, BCH -- $\Ka=100$};

\addplot [color=mycolor5,dashed,mark=triangle,mark options={solid,rotate=180}]
  table[row sep=crcr]{%
-28	  1\\
-25	  1\\
-23   0.96\\ 
-21   0.5350\\ 
-19   0.063\\ 
-17   8.59e-4\\ 
-16   5.3e-5\\ 
};
%\addlegendentry{Equivalent channel \eqref{eq:equivalent_channel0} and hard demapping, BCH -- $\Ka=200$};

\addplot [color=mycolor5,mark=o,mark options={solid,rotate=180}]
  table[row sep=crcr]{%
-28	  1\\
-25	  1\\
-23   0.96\\
-21   0.64\\
-19   0.10\\
-17   0.0013\\ %-17   0.008\\ 
-16   1.7e-4\\
};
%\addlegendentry{TBM, GN CPD and hard demapping, BCH code -- $\Ka=1$};

\addplot [color=mycolor5,mark=x,mark options={solid,rotate=180}]
  table[row sep=crcr]{%
-28	  1\\
-25	  1\\
-23   0.9906\\
-21   0.73\\
-19   0.1420\\
-17   0.0052\\
-16   3.5e-4\\
};
%\addlegendentry{TBM, GN CPD and hard demapping, BCH code -- $\Ka=100$};

0.9906    0.7300    0.1420

\addplot [color=mycolor5,mark=triangle,mark options={solid,rotate=180}]
  table[row sep=crcr]{%
-28	  1\\
-25	  1\\
-23   1\\
-21   0.9026\\
-19   0.2707\\
-17   0.0080\\
-16   7e-4\\
};
%\addlegendentry{TBM, GN CPD and hard demapping, BCH code -- BCH code $\Ka=200$};

%\addplot [color=mycolor5,dashed,mark=triangle,mark options={solid}]
%  table[row sep=crcr]{%
%-28	  1\\
%-27	  1\\
%-25	  1\\
%-23   0.9620\\
%-21   0.5070\\
%-19   0.06060\\ 
%-18   0.008850\\
%-17   0.000810\\ %
%};
%\addlegendentry{Model \eqref{eq:equivalent_channel0} - BCH code - $\Ka=200$};

	\end{axis}
\end{tikzpicture}%